\documentclass[12pt]{amsart}
\usepackage{amssymb, enumerate, mdwlist}

\evensidemargin\paperwidth
\advance\evensidemargin-\textwidth
\advance\oddsidemargin-1in
\evensidemargin\oddsidemargin

\topmargin\paperheight
\advance\topmargin-\textheight
\advance\topmargin-1in

\theoremstyle{plain}
\newtheorem{theorem}{Theorem}[section]
\theoremstyle{plain}

\theoremstyle{plain}
\newtheorem{lemma}[theorem]{Lemma}
\theoremstyle{plain}
\newtheorem{corollary}[theorem]{Corollary}
\theoremstyle{plain}

\theoremstyle{plain}
\newtheorem{definition}[theorem]{Definition}
\theoremstyle{plain}

\theoremstyle{remark}
\newtheorem{remark}[theorem]{Remark}
\theoremstyle{remark}
\newtheorem{example}[theorem]{Example}
\theoremstyle{remark}

\newcommand{\B}{\mathcal{B}}
\newcommand{\C}{\mathbb{C}}

\newcommand{\Hil}{\mathcal{H}}
\newcommand{\ii}{\mathbf{i}}
\newcommand{\kk}{\mathbf{k}}
\newcommand{\N}{\mathbb{N}}

\newcommand{\R}{\mathbb{R}}
\newcommand{\T}{\mathbb{T}}
\newcommand{\vv}{\mathbf{v}}

\newcommand{\wU}{\widehat{U}}
\newcommand{\xx}{\mathbf{x}}
\newcommand{\yy}{\mathbf{y}}
\newcommand{\Z}{\mathbb{Z}}

\title
[Categories of quantum walks]
{Categories of quantum walks}
\author{Hiroki Sako}
\address
{Faculty of engineering, Niigata University, Niigata 950-2181, Japan}
\email
{sako@eng.niigata-u.ac.jp}

\begin{document}

\begin{abstract}
We propose {\it categories} of $1$-dimensional and multi-dimensional quantum walks.
In the categories, {\it an object} is a quantum walk, and {\it a morphism} is an intertwining operator between two quantum walks.
The new framework enables us to discuss quantum walks in a unified way.
The purposes of this paper are the following:
\begin{enumerate}
\item
We reinterpret known results in our new framework.
\item
We show several new theorems.
For example,
it is proved that {\it every} space-homogeneous {\it time-periodic} analytic quantum walk on $\Z^d$ has a limit distribution of velocity for {\it every} initial unit vector.
Analyticity is a very weak condition.
\item
We ask whether there exists a continuous-time quantum walk $(V^{(t)})_{t \in \R}$ which realizes a given discrete-time quantum walk $U$.
Existence of $(V^{(t)})_{t \in \R}$ is equivalent to that of a $1$-parameter group of automorphisms
$(V^{(t)})_{t \in \R}$ from the object $U$ to $U$.
\end{enumerate}

\end{abstract}


\maketitle

\section{Introduction}
To see the motivation of this paper,
let us consider the following concrete $1$-dimensional quantum walk
acting on $\ell_2(\Z) \otimes \C^2$:
\[U = 
\left(
\begin{array}{cc}
a S & -b S\\
b & a
\end{array}
\right).
\]
In the above, let $S \colon \ell_2(\Z) \to \ell_2(\Z)$ be the bilateral shift,
$a, b$ be real numbers satisfying $a^2 + b^2 = 1$.
The above operator is unitary, and the dynamical system $(U^t)_{t \in \Z}$ forms a quantum walk.
Although our new project leads to several kinds of theorems,
let us first focus on the following corollary proved in Example \ref{example : existence of CTQW}: 
The discrete-time quantum walk 
$(U^t)_{t \in \Z}$ cannot be realized by a continuous-time quantum walk
.

In this paper, we propose a definition of quantum walks and define a new category consisting of quantum walks and intertwiners between them. This enables us to clearly understand such a result as above in an appropriate framework.

What do we need in order to understand the above result? 
First of all, we need to define quantum walks.
Unfortunately, many researchers study quantum walks without definitions.
The above result asks whether there exists a quantum walk satisfying some required condition, so it is unavoidable to clarify
what are quantum walks. 

We also propose categories of quantum walks.
This new framework enables us to study relationship between quantum walks.
Category is a modern tool for mathematics.
It consists of objects and morphisms.
An object is a point and a morphism is an arrow between two objects.
Objects are not necessarily sets and morphisms are not necessarily maps.
If every point has the identity morphism, and if there exists an associative
composition rule on morphisms, then
the whole system of objects and morphisms is said to form a category.
See \cite{EilenbergMacLane} for the general theory of categories.
In our new category, a quantum walk is an object (a point), and an intertwiner between two quantum walks is a morphism (an arrow).
An intertwiner $W \colon \Hil_1 \to \Hil_2$ between two quantum walks $(U_1 \curvearrowright \Hil_1)$, $(U_2 \curvearrowright \Hil_2)$ 
is an operator satisfying $W U_1 = U_2 W$.
See Section \ref{section: category of QWs}.

Using our new category, we study relationship between given two quantum walks $(U_1 \curvearrowright \Hil_1)$, $(U_2 \curvearrowright \Hil_2)$.
If there exists an isometry $W \colon \Hil_1 \to \Hil_2$ which is an intertwiner from $U_1$ to $U_2$, and if some dynamical system can be described by $U_1$, then the dynamical system is also realized by $U_2$.
We will give several concrete examples in Section
\ref{section: morphisms between 1-dim homogeneous QWs}.

This category is useful, when we study whether a given quantum walk $(U \curvearrowright \Hil)$ can be realized by a continuous-time quantum walk $(V^{(t)})_{t \in \R}$. The condition is equivalent to the condition that there exists $(V^{(t)})_{t \in \R} \curvearrowright \Hil$ satisfying $V^{(1)} = U$.
For every real number $t$, the operator $V^{(t)}$ satisfies
\[V^{(t)} U = V^{(t + 1)} = V^{(1 + t)} = U V^{(t)}.\]
This means that $V^{(t)}$ is a morphism from the object $(U \curvearrowright \Hil)$ to $(U \curvearrowright \Hil)$.
Thus, realizability by a continuous-time quantum walk
is equivalent to existence of a $1$-parameter family of automorphisms at the object $(U \curvearrowright \Hil)$.
Once we capture
the whole structure of the category,
we can determine whether the condition holds or not.
%

\subsection*{Contents of this paper}

\paragraph{\bf New results:}
\begin{itemize}
\item
We define four kinds of regularity for quantum walks:
having finite propagation, analyticity, smoothness, uniformity.
In preceding research, almost all the quantum walks have finite propagation.
We prove in Theorem \ref{theorem: 4 kinds of regularity} that the condition for having finite propagation implies analyticity.
We also have implications ``analyticity $\Rightarrow$ smoothness'', ``smoothness $\Rightarrow$ uniformity''.
\item
For every multi-dimensional discrete-time smooth quantum walk,
and for {\it every initial unit vector},
the support of the distribution of velocity is asymptotically compact.
See Theorem \ref{theorem: asymptotically having a compact support} for the statement.
In the proceeding papers, researchers mainly study the case that the initial unit vector is {\it finitely supported}.
\item
In Theorem \ref{theorem: convergence and intertwiner},
we prove that if two smooth quantum walks have a good intertwiner, and one of them has a limit distribution of velocity, then the other walk also has.
\item
In Theorem
\ref{theorem: convergence of time-periodic homogeneous walk},
we prove that
{\it every} space-homogeneous {\it time-periodic} discrete-time analytic quantum walk has a weak limit distribution of velocity for {\it every initial unit vector}.

\end{itemize}
\paragraph{\bf Reinterpretations of known arguments:}
\begin{itemize}
\item
Many researchers have already utilized Fourier analysis for the study of quantum walks.
We clarify in this paper that the inverse Fourier transform of a walk is also an object. See Section \ref{section: general theory of multi-dim QWs}.
\item
The paper \cite{SaigoSako} has already explained that model quantum walks play an important role in the study of $1$-dimensional homogeneous quantum walks.
In our new framework, we can make it easy to manipulate the model quantum walks. 
See Section \ref{section: structure theorems of $1$-dim HQW}
and Section \ref{section: morphisms between 1-dim homogeneous QWs}.
\item
For a given discrete-time quantum walk $U$, existence of a continuous quantum walk $(V^{(t)})_{t \in \R}$ satisfying that $V^{(1)} = U$ is equivalent to existence of $1$-parameter auto morphisms $(V^{(t)})_{t \in \R}$ on the object $U$. We explain how one can determine this type of realizability for $1$-dimensional discrete-time homogeneous analytic quantum walks $U$. See Section \ref{section: realizability by CTQW}.
\end{itemize}

\section{The category of Hilbert spaces with coordinates}

\subsection{Objects}

We consider Hilbert spaces which is associated to the Euclidean space $\R^d$.
We can also define the same kind of notion for topological spaces other than $\R^d$.
\begin{definition}\label{definition: d-dim Hilbert space}
Let $\Hil$ be a complex Hilbert space, and let $d$ be a natural number.
Let $E$ is a map which maps every Borel subset of $\R^d$ to an orthogonal projection in $\B(\Hil)$. 
A pair $(\Hil, E)$ is called {\rm a Hilbert space with a $d$-dimensional coordinate system}, if
$E$ is countably additive and $E(\R^d)$ is the identity operator on $\Hil$.
The map $E$ is called {\rm a spectral measure} on $\R^d$ or a coordinate system of $\Hil$. 
\end{definition}

Consider the case that $d$-tuple $(h_1, \cdots, h_d)$ of self-adjoint operators (or observables) on $\Hil$ is given. 
If these operators have a common core $\mathcal{D} \subset \Hil$, and if these operators are commutative on $\mathcal{D}$, then the $d$-tuple is expressed by a spectral measure $E$.
The pair $(\Hil, E)$ is a Hilbert space with a $d$-dimensional coordinate system.
We can also describe $(\Hil, E)$ as $(\Hil, h_1, \cdots, h_d)$.

Given a Hilbert space $(\Hil, E)$ with a $d$-dimensional coordinate system
and a vector $\xi$, we obtain a measure
$\mu_\xi(\Omega) = \langle E(\Omega) \xi, \xi \rangle$
on $\R^d$.
The volume $\mu_\xi(\R^d)$ of $\R^d$ is equal to $\|\xi\|^2$.
If $\xi$ is a unit vector, then $\mu_\xi$ is a probability measure.
The system $(\Hil, E)$ associates a unit vector $\xi$ to a probability measure $\mu_\xi$.
We sometimes call $\mathrm{supp}(\mu_\xi)$ the support of $\xi$. 
If the self-adjoint operators $(h_1, \cdots, h_d)$ correspond to observables of position, 
then the measure $\mu_\xi$ is the distribution of position of particles whose states are expressed by $\xi$.

\begin{example}\label{example: standard coordinate system}
To the Hilbert space $\ell_2 (\Z^d) \otimes \C^n$, we implicitly associate the following form of spectral measure $E$ on $\R^d$:
\[E \colon \Omega \mapsto (\textrm{The\ orthogonal\ projection\ onto\ } \ell_2(\Omega \cap \Z^d) \otimes \C^n).\]
In this paper, we call $E$ the standard coordinate system of $\ell_2 (\Z^d) \otimes \C^n$.
For a vector $\xi = \sum_{\xx \in \Z^d} \delta_\xx \otimes \xi_\xx$,
the measure $\mu_\xi$ is equal to the following:
\[\mu_\xi(\Omega) 
= \langle E(\Omega) \xi, \xi \rangle
= \left \langle  
\sum_{\xx \in \Omega \cap \Z^d} \delta_\xx \otimes \xi_\xx, 
\sum_{\xx \in\Z^d} \delta_\xx \otimes \xi_\xx 
\right\rangle^2
= \sum_{\xx \in \Omega \cap \Z^d} \|\xi_\xx\|^2.\]

In the study of quantum walks, many researchers have been using this kind of Hilbert spaces.
We can widen the framework of quantum walks by Definition \ref{definition: d-dim Hilbert space}.
\end{example}

\begin{example}
Let us denote by $\T_{2 \pi}$ the torus $\R/ 2 \pi \Z$.
For every open interval in $\T_{2 \pi}$, we
have a coordinate $k$ arising from $\R$.
Let us consider the Hilbert space $L^2(\T_{2 \pi})$. 
The differential operator $\frac{1}{\ii} \frac{d}{d k}$ is an unbounded self-adjoint operator on $L^2(\T_{2 \pi})$. The support of the spectral decomposition of 
$\frac{1}{\ii} \frac{d}{d k}$ is $\Z$. The spectral measure defines a $1$-dimensional coordinate system of $L^2(\T_{2 \pi})$.
Recycling notations, we simply denote by $\left( L^2(\T_{2 \pi}), \frac{1}{\ii} \frac{d}{d k} \right)$ this Hilbert space with the $1$-dimensional coordinate system.
\end{example}

\begin{example}
Let $(\Hil_1, E_1)$ and $(\Hil_2, E_2)$ be Hilbert spaces with $d$-dimensional coordinate systems.
We define the spectral measure $E_1 \oplus E_2$ by
$\R^d \supset \Omega \mapsto E_1(\Omega) \oplus E_2(\Omega)$.
The direct sum $(\Hil_1 \oplus \Hil_2, E_1 \oplus E_2)$ is also a Hilbert space with a $d$-dimensional coordinate system.
\end{example}

In the next definition, we define four kinds of regularity for vectors in $\Hil$.
We make use of the following unitary representation of $\R^d \cong \widehat \R^d$:
\[v_E(\kk) = \int_{\xx \in \R^d} \exp(\ii \xx \cdot \kk) E(\xx) d \xx.\]
In the above equation, $\ii$ stands for the imaginary unit, and $ \xx \cdot \kk$ stands for the standard inner product of $\R^d$.
In the case that $\xx$ stands for position, $\kk$ stands for a wavenumber.

\begin{definition}
Let $(\Hil, E)$ be a Hilbert space with a $d$-dimensional coordinate system.
Let $\xi$ be a vector in $\Hil$.
\begin{enumerate}
\item
It is said that the support of $\xi$ is compact,
if there exists a compact subset $\Omega$ of $\R^d$
such that
$\xi \in E(\Omega) \Hil$.
\item
The vector $\xi$ is said to be {\rm analytic},
if the mapping $\R^d \ni \kk \mapsto v_E(\kk) \xi \in \Hil$
can be extended to a holomorphic map defined in a domain in $\C^d$.
\item
The vector $\xi$ is said to be {\rm smooth},
if the mapping $\R^d \ni \kk \mapsto v_E(\kk) \xi \in \Hil$
is smooth.
\item
The vector $\xi$ is said to be {\rm uniform},
if the mapping $\R^d \ni \kk \mapsto v_E(\kk) \xi \in \Hil$
is continuous.
\end{enumerate}
\end{definition}
To discuss differentiability of a map from $\R^d$ or a domain of $\C^d$,
for the target space,
we have only to require linearity and a compete metric
on it. 
Since $\Hil$ is a Banach space, the above definition makes sense.

\begin{lemma}\label{lemma: smoothness for vectors}
A $\xi$ in $\Hil$ is smooth,
if and only if for every real-valued monomial function $p$,
$\xi$ is in the domain of $\int_{\R^d} p(\xx) d E(\xx)$.
\end{lemma}

\begin{proof}
Express $\xx \in \R^d$ by $(x_1, \cdots, x_j, \cdots, x_d)$.
Express $\kk \in \R^d$ by $(k_1, \cdots, k_j, \cdots, k_d)$.
Define a self-adjoint operator $D_j$ by $\int_{\xx \in \R^d} x_j d E(\xx)$.
It suffices to show that
the mapping $\R^d \ni \kk \mapsto v_E(\kk) \xi \in \Hil$
has the partial derivative $\frac{\partial}{\partial k_j} v_E(\kk) \xi$,
if and only if $\xi$ is in the domain of $D_j$.

For every $\xi$ in $\Hil$,
we calculate the norm of the average of the change as follows:
\begin{eqnarray*}
&&
\left\|
\frac{
v_E(k_1, \cdots, k_{j - 1}, k_j + h, k_{j +1}, \cdots, k_d) \xi - 
v_E(k_1, \cdots, k_d) \xi}{h}
\right\|^2
\\
&=&
\left\|
\int_{\xx \in \R^d} \frac{\exp(\ii \xx \cdot \kk)(\exp(\ii h x_j) - 1)}{h}
d E(\xx) \xi
\right\|^2
\\
&=&
\int_{\xx \in \R^d} \left| \frac{\exp(\ii h x_j) - 1}{h} \right|^2
d \left\| E(\xx) \xi \right\|^2.
\end{eqnarray*}
As $h$ tends to $0$, by Beppo Levi's monotone convergence theorem for Lebesgue integral,
the right hand side converges to
$\int_{\xx \in \R^d} x_j^2
d \left\| E(\xx) \xi \right\|^2 \in [0, \infty]$.
Suppose that the partial derivative $\frac{\partial}{\partial k_j} v_E(\kk) \xi$
exists.
Since the left hand side converges to 
the real number $\left\| \frac{\partial}{\partial k_j} v_E(\kk) \xi \right\|^2 < \infty$,
the quantity $\int_{\xx \in \R^d} x_j^2
d \left\| E(\xx) \xi \right\|^2$ is finite in this case.
It follows that $\xi$ is in the domain of $D_j$.

Conversely suppose that $\xi$ is in the domain of $D_j$.
In this case, $x_j^2$ is integrable with respect to the probability measure $\Omega \mapsto \|E(\Omega) \xi\|^2$.
The average of the change satisfies the following:
\begin{eqnarray*}
&&
\left\|
\frac{
v_E(k_1, \cdots, k_{j - 1}, k_j + h, k_{j +1}, \cdots, k_d) \xi - 
v_E(k_1, \cdots, k_d) \xi}{h} - \ii D_j \xi
\right\|^2
\\
&=&
\int_{\xx \in \R^d} \left| \frac{\exp(\ii h x_j) - 1}{h} - \ii x_j \right|^2
d \left\| E(\xx) \xi \right\|^2.
\end{eqnarray*}
The integrand is dominated by the integrable function $2 x_j^2$ and converges to $0$ as $h$ tends to $0$.
By Lebesgue's dominated convergence theorem,
as $h$ tends to $0$,
the average of the change converges to $\ii D_j \xi$.
\end{proof}

\subsection{Morphisms}

\begin{definition}
Let $(\Hil_1, E_1)$ and $(\Hil_2, E_2)$ be Hilbert spaces with $d$-dimensional coordinate systems. A bounded linear operator $W$ from $\Hil_1$ to $\Hil_2$ is called {\rm a morphism} from the object $(\Hil_1, E_1)$ to the object $(\Hil_2, E_2)$ 
\end{definition}

In the next definition, we are going to define four kinds of regularity on morphisms.

\begin{definition}\label{definition: regularity}
Let $W$ be a morphism from $(\Hil_1, E_1)$ to $(\Hil_2, E_2)$.
Let $v_j = v_{E_j}$ be the unitary representation of $\R^d = \widehat\R^d$ associated to  $(\Hil_j, E_j)$.
For every $\kk \in \R^d$, define $W_{\kk} \in \B(\Hil_2 \leftarrow \Hil_1)$ by
$W_\kk 
= v_2(\kk) \cdot W \cdot v_1(- \kk).$
We define four kinds of regularity as follows:
\begin{enumerate}
\item
The operator $W$ is called {\rm a morphism with finite propagation},
if the following condition holds:
there exists a constant $R = R_W$ such 
that for every Borel subset $A_1, A_2$ of $\R^d$,
if $\mathrm{dist}(A_1, A_2) > R$, then $E_2(A_2) W E_1(A_1) = 0$.
\item
The operator $W$ is called an {\rm analytic} morphism,
if the following condition holds:
there exists a holomorphic map from a domain in $\C^d$ including $\R^d$ to $\B(\Hil)$ which extends the mapping
\[\R^d \ni \kk 
\mapsto 
W_\kk
\in \B(\Hil_2 \leftarrow \Hil_1).
\]
\item
The operator $W$ is called a {\rm smooth} morphism,
if the mapping
\[\R^d \ni \kk 
\mapsto W_\kk
\in \B(\Hil_2 \leftarrow \Hil_1)
\]
is smooth.
\item
The operator $W$ is called a {\rm uniform} morphism,
if the mapping
\[\R^d \ni \kk 
\mapsto W_\kk
\in \B(\Hil_2 \leftarrow \Hil_1)
\]
is continuous with respect to the operator norm.
\end{enumerate}
\end{definition}

Since $\B(\Hil_2 \leftarrow \Hil_1)$ is a Banach space, the above definition makes sense.

It is not hard to show the following claims:
\begin{itemize}
\item
If $W \in \B(\Hil_2 \leftarrow \Hil_1)$ has finite propagation, and if the support of $\xi \in \Hil_1$ is compact, then the support of $W \xi \in \Hil_2$ is also compact.
\item
A composition of two morphisms with finite propagation has finite propagation.
\item
If $W \in \B(\Hil_2 \leftarrow \Hil_1)$ is analytic (smooth, or uniform), and if $\xi \in \Hil_1$ is analytic (smooth, or uniform, respectively), then the support of $W \xi \in \Hil_2$ is also analytic (smooth, or uniform, respectively).
\item
A composition of two analytic (smooth, or uniform) morphisms
is analytic (smooth, or uniform, respectively).
\end{itemize}

\begin{example}
Denote by $\mathcal{F} \colon L^2(\T_{2 \pi}) \to \ell_2(\Z)$ the Fourier transform.
The Hilbert space $L^2(\T_{2 \pi})$ has the $1$-dimensional coordinate system by the eigenvalue decomposition of the self-adjoint operator $\frac{1}{\ii} \frac{d}{d k}$.
The Hilbert space $\ell_2(\Z)$ has the standard $1$-dimensional coordinate system.
Suppose that a vector $\xi \in L^2(\T_{2 \pi})$ is located at $x \in \Z$, or equivalently, suppose that $\xi(k)$ is a scalar multiple of $\exp(\ii k x)$.
Then the support of $\mathcal{F} \xi \in \ell_2(\Z)$ is $\{x\}$.
This means that the operator $\mathcal{F}$ does not change the support of vectors. It follows that $\mathcal{F}$ has finite propagation.
\end{example}

\begin{example}
Let $f$ be a bounded Borel function on $\T_{2 \pi}$.
Then we obtain a bounded linear operator $M[f] \colon L^2(\T_{2 \pi}) \to L^2(\T_{2 \pi})$ defined by $\xi \mapsto f \xi$.
We call it the multiplication operator of $f$.
As a morphism from $(L^2(\T_{2 \pi}), \frac{1}{\ii}\frac{d}{dk})$ to itself,
$M[f]$ is analytic (smooth, or uniform), if and only if $f$ is analytic (smooth, or continuous, respectively).
The key for the proof is that the unitary representation $(v(k))_{k \in \R}$
on $(L^2(\T_{2 \pi}), \frac{1}{\ii}\frac{d}{dk})$ is given by the rotation on the torus $\T_{2 \pi}$.
The morphism $M[f]$ has finite propagation, if and only if support of the Fourier expansion of $f$ is a sum of finitely many terms.
\end{example}

\begin{example}
The bounded linear operator 
$\mathcal{F} M[f] \mathcal{F}^{-1} \colon \ell_2(\Z) \to \ell_2(\Z)$
is equal to the convolution operator given by the Fourier coefficients of $f$.
As a morphism from $\ell_2(\Z)$ to itself,
$\mathcal{F} M[f] \mathcal{F}^{-1}$ is analytic (smooth, or uniform), if and only if $f$ is analytic (smooth, or continuous, respectively).
The morphism $\mathcal{F} M[f] \mathcal{F}^{-1}$ has finite propagation, if and only if support of the Fourier coefficients of $f$ is finite.
\end{example}

\begin{theorem}\label{theorem: 4 kinds of regularity}
Let $W$ be a morphism from $(\Hil_1, E_1)$ to $(\Hil_2, E_2)$.
\begin{enumerate}
\item
If the morphism $W$ has finite propagation, then $W$ is analytic.
\item
If the morphism $W$ is analytic, then $W$ is smooth.
\item
If the morphism $W$ is smooth, then $W$ is uniform.
\end{enumerate}
\end{theorem}

\begin{proof}
The second and the third items are trivial. 
We prove the first one.

For $\xx = (x_1, \cdots, x_d) \in \Z^d$, define $C_\xx \subset \R^d$ by the direct product of the semi-open intervals
\[[x_1, x_1 + 1) \times \cdots \times [x_d, x_d + 1).\]
For every morphism $W$ with finite propagation, we have
\begin{eqnarray*}
W 
&=& 
\left( \sum_{\yy \in \Z^d} E_2(C_{\yy}) \right) 
W 
\left( \sum_{\xx \in \Z^d} E_1(C_{\xx}) \right)\\
&=& 
\sum_{\yy \in \Z^d} \sum_{\xx \in \Z^d} E_2(C_{\xx + \yy}) 
W 
E_1(C_{\xx}).
\end{eqnarray*}
Since $W$ has finite propagation, there exists a finite subset $F$ of $\Z^d$ such that 
for every $\yy \in \Z^d \setminus F$, $E_2(C_{\xx + \yy}) W E_1(C_{\xx}) = 0$.
It follows that 
\begin{eqnarray*}
W 
&=& 
\sum_{\yy \in F} \left( \sum_{\xx \in \Z^d} E_2(C_{\xx + \yy}) 
W 
E_1(C_{\xx}) \right).
\end{eqnarray*}
Define $W^{(\yy)}$ by $\sum_{\xx \in \Z^d} E_2(C_{\xx + \yy}) 
W 
E_1(C_{\xx})$.
It suffices to show that for every fixed $\yy \in F$, 
the map
\[\kk \mapsto v_2(\kk) \cdot W^{(\yy)} \cdot v_1(- \kk)\]
can be extended to a holomorphic map defined on $\C^d$.

For every $\xx \in \Z^d$, we have
\begin{eqnarray*}
&&v_2(\kk) \cdot E_2(C_{\xx + \yy}) W E_1(C_{\xx}) \cdot v_1(- \kk)\\
&=&
\int_{\mathbf{t} \in C_{\xx + \yy}} \exp(\ii \mathbf{t} \cdot \kk) d E_2 (\mathbf{t})
\cdot W \cdot  
\int_{\mathbf{s} \in C_{\xx}} \exp(- \ii \mathbf{s} \cdot \kk) d E_1 (\mathbf{s})\\
&=&
\int_{\mathbf{t} \in C_{\mathbf{0}}} \exp(\ii (\mathbf{t} + \xx + \yy) \cdot \kk) 
d E_2 (\mathbf{t} + \xx + \yy)
\cdot W \cdot  
\int_{\mathbf{s} \in C_{\mathbf{0}}} \exp(- \ii (\mathbf{s} + \xx) \cdot \kk) 
d E_1 (\mathbf{s} + \xx)\\
&=&
 \exp(\ii \yy \cdot \kk)
\int_{\mathbf{t} \in C_{\mathbf{0}}} \exp(\ii \mathbf{t} \cdot \kk) 
d E_2 (\mathbf{t} + \xx + \yy)
\cdot W \cdot  
\int_{\mathbf{s} \in C_{\mathbf{0}}} \exp(- \ii \mathbf{s} \cdot \kk) 
d E_1 (\mathbf{s} + \xx).
\end{eqnarray*}
Denote by $W^{(\yy + \xx, \xx)}_\kk$ the right hand side. 
Note that for every $\xx, \yy \in \Z^d$, the three mappings 
\[\exp(\ii \yy \cdot \kk),
\quad
\int_{\mathbf{t} \in C_{\mathbf{0}}} \exp(\ii \mathbf{t} \cdot \kk) 
d E_2 (\mathbf{t} + \xx + \yy),
\quad
\int_{\mathbf{s} \in C_{\mathbf{0}}} \exp(- \ii \mathbf{s} \cdot \kk) 
d E_1 (\mathbf{s} + \xx).\]
defined on $\kk \in \C^d$, are holomorphic.
Therefore
the mapping
$\kk \mapsto W^{(\yy + \xx, \xx)}_\kk$  is a holomorphic map defined on $\C^d$.
Also note that for every $\yy \in F$, there exists a positive function $\rho(\kk), (\kk \in \C^d)$ independent of $\xx$ such that 
\begin{itemize}
\item
The norm of $W^{(\yy + \xx, \xx)}_\kk$ is at most $\rho(\kk)$,
\item
The norms of the partial derivatives of $W^{(\yy + \xx, \xx)}_\kk$ with respect to $\kk = (k_1, \cdots, k_d) \in \C^d$ are at most $\rho(\kk)$.
\end{itemize}
For every $\kk \in \C^d$, the operators $\{W^{(\yy + \xx, \xx)}_\kk\}_{\xx \in \Z^d}$
have mutually orthogonal images, and the orthogonal complements of the kernels of them are also mutually perpendicular.
This means that the 
\[\kk \mapsto v_2(\kk) \cdot W^{(\yy)} \cdot v_1(- \kk) =
\sum_{\xx \in \Z^d} W^{(\yy + \xx, \xx)}_\kk 
\]
can be extended to a holomorphic mapping defined on $\kk \in \C^d$.
\end{proof}

In the study of quantum walks, researchers mainly pay attention on unitary operators with finite propagation, and call them quantum walks.
For the rest of this paper, we study quantum walks using knowledge of functional analysis.
Instead of the condition for having finite propagation,
we make use of three other conditions as functional analytic necessary conditions.

\begin{definition}
A triple $(\Hil, E, (\rho(\xx))_{\xx \in \Z^d})$
is called a Hilbert space with {\rm a right regular representation} of $\Z^d$,
if 
$(\Hil, E)$ is a Hilbert space with a $d$-dimensional  coordinate system and $\rho$ is a unitary representation of $\Z^d$ on $\Hil$ and if the following condition holds:
for every Borel subset $A \subset \R^d$ and for every $\xx \in \Z^d$, $\rho(\xx) E(A) \rho(- \xx) = E(A - \xx)$.
\end{definition}

\section{The categories of quantum walks}
\label{section: category of QWs}

\subsection{Objects}

\begin{definition}
A triple $(\Hil, (U^t)_{t \in \Z}, E)$
is called {\rm a $d$-dimensional quantum walk},
if
$(\Hil, E)$ is a Hilbert space with a $d$-dimensional  coordinate system and $(U^t)_{t \in \Z}$ is a unitary representation of $\Z$ on $\Hil$.
\end{definition}
To emphasize that $U$ is not a representation of $\R$ but of $\Z$, we call it a discrete-time quantum walk.

\begin{definition}
A triple $(\Hil, (U^{(t)})_{t \in \R}, E)$
is called {\rm a $d$-dimensional continuous-time quantum walk},
if
$(\Hil, E)$ is a Hilbert space with a $d$-dimensional coordinate system and $(U^{(t)})_{t \in \R}$ is a unitary representation of $\R$ on $\Hil$
which is continuous with respect to the strong operator topology.
\end{definition}

We can also define quantum walks on Hilbert spaces which is associated to topological spaces other than $\R^d$.
We also define four kinds of regularity for $d$-dimensional quantum walks.

\begin{definition}
Let $(\Hil, (U^{(t)})_{t}, E)$
be a discrete-time or a continuous-time quantum walk.
\begin{enumerate}
\item
The quantum walk is said {\rm to have finite propagation},
if for every $t$, the unitary operator $U^{(t)}$ has finite propagation as a morphism from 
$(\Hil, E)$ to $(\Hil, E)$.
\item
The quantum walk is said to be {\rm analytic (smooth, or uniform)},
if for every $t$, the unitary operator $U^{(t)}$ is analytic (smooth, or uniform, respectively) as a morphism from $(\Hil, E)$ to $(\Hil, E)$.
\end{enumerate}
\end{definition}

For a continuous-time quantum walk,
we may define some notion of regularity
using the generator of the representation $(U^{(t)})_{t \in \R}$.
It seems that such kind of regularity is stronger than above.

\begin{example}
Let $E$ be the standard coordinate system of $\ell_2(\Z) \otimes \C^2$ defined in Example \ref{example: standard coordinate system}.
Define a unitary operator $U$ by
\[U = 
\left(
\begin{array}{cc}
a S & -b S\\
b S^{-1}& a S^{-1}
\end{array}
\right).
\]
In the above, $S \colon \ell_2(\Z) \to \ell_2(\Z)$ be the bilateral shift,
$a, b$ are real numbers satisfying $a^2 + b^2 = 1$.
The dynamical system $(U^t)_{t \in \Z}$ forms a $1$-dimensional discrete-time quantum walk acting on the Hilbert space $\ell_2(\Z) \otimes \C^2$ with the standard coordinate system.
This walk has finite propagation.
\end{example}

\begin{example}
The $1$-parameter family of unitary operators
$(\exp(\ii t (S + S^{-1})))_{t \in \R}$
gives a $1$-dimensional continuous-time quantum walk acting on the Hilbert space $\ell_2(\Z)$. This quantum walk is analytic, although it does not have finite propagation.
\end{example}

\begin{example}
The unitary operator 
$\left(
\begin{array}{cc}
a S & -b S\\
b & a
\end{array}
\right)$
is a $1$-dimensional discrete-time quantum walk acting on the Hilbert space $\ell_2(\Z) \otimes \C^2$ with finite propagation.
\end{example}

\begin{example}
The unitary operator 
$\dfrac{1}{2} \left(
\begin{array}{cc}
S^3 + S & S - S^{-1}\\
S - S^{-1} & S^{-1} + S^{-3}
\end{array}
\right)$
is a $1$-dimensional discrete-time quantum walk acting on the Hilbert space $\ell_2(\Z) \otimes \C^2$ with finite propagation.
\end{example}

\begin{example}\label{example: 4-state Grover walk}.
The unitary operator 
$\dfrac{1}{2}
\left(
\begin{array}{ccccc}
S^{-3} &    0    &   0  & 0       \\
0       & S^{-1} &   0  & 0 \\
0       &    0     &  S  & 0       \\
0       &    0     &  0  & S^{3}
\end{array}
\right)
\left(
\begin{array}{ccccc}
1 & -1 & -1 & -1 \\
-1 & 1 & -1 & -1 \\
-1 & -1 & 1 & -1 \\
-1 & -1 & -1 & 1 \\
\end{array}
\right)$
is a $1$-dimensional discrete-time quantum walk acting on the Hilbert space $\ell_2(\Z) \otimes \C^4$ with finite propagation.
We will call it the $1$-dimensional $4$-state Grover walk.
\end{example}

\begin{example}
Let
$
u_x =
\left(
\begin{array}{cc}
a_x & b_x\\
c_x & d_x
\end{array}
\right), x \in \Z
$
be a two-sided infinite sequence of complex unitary matrices.
Define a unitary operator $U$ on $\ell_2(\Z) \otimes \C^2$ as follows:
\begin{eqnarray*}
U (\delta_x \otimes \delta_1) &=& a_x \delta_{x - 1} \otimes \delta_1 + c_x \delta_{x + 1} \otimes \delta_2,\\
U (\delta_x \otimes \delta_2) &=& b_x \delta_{x - 1} \otimes \delta_1 + d_x \delta_{x + 1} \otimes \delta_2, \quad x \in \Z.
\end{eqnarray*}
Then $(U^t)_{t \in \Z}$ gives a $1$-dimensional discrete-time quantum walk with finite propagation.
\end{example}

\begin{example}
Let $E$ be the $2$-dimensional standard coordinate system of $\ell_2(\Z^2) \otimes \C^2$ defined in Example \ref{example: standard coordinate system}.
Denote by $S_{\rm r} \colon \ell_2(\Z^2) \to \ell_2(\Z^2)$ be the bilateral shift to the right.
Denote by $S_{\rm u} \colon \ell_2(\Z^2) \to \ell_2(\Z^2)$ be the upward bilateral shift.
The unitary operator
\[\dfrac{1}{2}
\left(
\begin{array}{ccccc}
S_{\rm r} &    0    &   0  & 0       \\
0       & S_{\rm r}^{-1} &   0  & 0 \\
0       &    0     &  S_{\rm u}  & 0       \\
0       &    0     &  0  & S_{\rm u}^{-1}
\end{array}
\right)
\left(
\begin{array}{ccccc}
1 & -1 & -1 & -1 \\
-1 & 1 & -1 & -1 \\
-1 & -1 & 1 & -1 \\
-1 & -1 & -1 & 1 \\
\end{array}
\right)
\]
defines a discrete-time $2$-dimensional quantum walks with finite propagation acting on $\ell_2(\Z^2) \otimes \C^4$.
\end{example}

\begin{example}
The unitary operator
$U
=
\dfrac{1}{2}
\left(
\begin{array}{cc}
S_{\rm r} + S_{\rm u} & - S_{\rm r}^{-1} + S_{\rm u}^{-1} \\ 
S_{\rm r} - S_{\rm u} & S_{\rm r}^{-1} + S_{\rm u}^{-1}
\end{array}
\right)
$
defines a discrete-time $2$-dimensional quantum walks with finite propagation acting on $\ell_2(\Z^2) \otimes \C^2$.
\end{example}

\begin{example}\label{example: direct sum of QWs}
Let $(\Hil_1, (U_1^t)_{t \in \Z}, E_1)$
and 
$(\Hil_2, (U_2^t)_{t \in \Z}, E_2)$
be $d$-dimensional quantum walks.
Then the direct sum
$(\Hil_1 \oplus \Hil_2, (U_1^t \oplus U_2^t)_{t \in \Z}, E_1 \oplus E_2)$
is also a $d$-dimensional quantum walk.
The four kinds of regularity are preserved under this manipulation.
\end{example}

Every quantum walk $(\Hil, (U^{(t)})_{t \in \Z}, E)$ forms an object in our new category.

\subsection{Morphisms in the category}

\begin{definition}
Let $(\Hil_1, (U_1^t)_{t \in \Z}, E_1)$
and 
$(\Hil_2, (U_2^t)_{t \in \Z}, E_2)$
be $d$-dimensional discrete-time quantum walks.
Let $T$ be a bounded linear operator from $\Hil_1$ to $\Hil_2$.
The operator $T$ is called {\rm an intertwiner (or a morphism)} from
$(\Hil_1, (U_1^t)_{t \in \Z}, E_1)$
to 
$(\Hil_2, (U_2^t)_{t \in \Z}, E_2)$,
if $T U_1 = U_2 T$ holds.
\end{definition}

The identity operator on $\Hil_1$ is an intertwiner from
$(\Hil_1, (U_1^t)_{t \in \Z}, E_1)$
to 
$(\Hil_1, (U_1^t)_{t \in \Z}, E_1)$
Let $T_1$ be an intertwiner (or a morphism) from
$(\Hil_1, (U_1^t)_{t \in \Z}, E_1)$
to 
$(\Hil_2, (U_2^t)_{t \in \Z}, E_2)$.
Let $T_2$ be an intertwiner (or a morphism) from
$(\Hil_2, (U_2^t)_{t \in \Z}, E_2)$
to 
$(\Hil_3, (U_3^t)_{t \in \Z}, E_3)$.
Then $T_2 T_1$ is an intertwining operator (or a morphism) from
$(\Hil_1, (U_1^t)_{t \in \Z}, E_1)$
to 
$(\Hil_3, (U_3^t)_{t \in \Z}, E_3)$.
It follows that $d$-dimensional discrete-time quantum walks and morphisms between them form a category.

It is straightforward to show the following lemma.

\begin{lemma}
If $T$ is an intertwiner(or a morphism) from
$(\Hil_1, (U_1^t)_{t \in \Z}, E_1)$
to 
$(\Hil_2, (U_2^t)_{t \in \Z}, E_2)$,
then
\begin{itemize}
\item
$T^*$ is an intertwiner from $(\Hil_2, (U_2^t)_{t \in \Z}, E_2)$ to $(\Hil_1, (U_1^t)_{t \in \Z}, E_1)$,
\item
$T^* T$ and $|T|$ is an intertwiner from 
$(\Hil_1, (U_1^t)_{t \in \Z}, E_1)$ to $(\Hil_1, (U_1^t)_{t \in \Z}, E_1)$.
\item
The partial isometry $V$ in the polar decomposition $T = |T| V$ is an intertwiner from
$(\Hil_1, (U_1^t)_{t \in \Z}, E_1)$
to 
$(\Hil_2, (U_2^t)_{t \in \Z}, E_2)$.
\end{itemize}
\end{lemma}

In the above definition, 
morphisms have no relation with the coordinates $E_j$ of the Hilbert spaces $\Hil_j$.
For the rest of this paper, we focus on morphisms satisfying regularity conditions
defined in the following definition.

\begin{definition}
Let $T$ be an intertwiner (or a morphism) from
$(\Hil_1, (U_1^t)_{t \in \Z}, E_1)$
to 
$(\Hil_2, (U_2^t)_{t \in \Z}, E_2)$.
\begin{enumerate}
\item
The intertwining operator $T$ is said {\rm to have finite propagation},
if $T$ has finite propagation as a morphism from $(\Hil_1, E_1)$ to $(\Hil_2, E_2)$.
\item
The intertwining operator $T$ is said to be {\rm analytic (smooth, or uniform)},
if $T$ is analytic (smooth, or uniform, respectively) morphism from $(\Hil_1, E_1)$ to $(\Hil_2, E_2)$.
\end{enumerate}
\end{definition}

\subsection{Our new categories}
For every $d \in \N$, we have
the category $\mathcal{C}_{d, \textrm{DTQW}}$ whose objects are $d$-dimensional discrete-time quantum walks and morphisms are intertwiners.
We often erase objects and morphisms which do not satisfy regularity.
We have defined the four families of objects and four kinds of morphisms
Thus we obtain sixteen categories.
We will mainly pay attention on the category consisting of analytic quantum walks
and of smooth intertwiners.
We can not find geometric meaning of 
quantum walks without regularity nor
intertwiners without regularity

\subsection{Similarity between two $d$-dimensional quantum walks}

\begin{definition}
Let $U_1 = (\Hil_1, (U_1^{(t)})_{t}, E_1)$, $U_2 = (\Hil_2, (U_2^{(t)})_{t}, E_2)$ be two discrete-time or continuous-time $d$-dimensional quantum walks.
If there exists a smooth invertible intertwiner $T \colon \Hil_1 \to \Hil_2$ between $U_1$ and $U_2$ (or equivalently, a smooth isomorphism between two objectis $U_1$ and $U_2$), then we say that these two quantum walks are similar.
\end{definition}

It is not hard to see that similarity defines an equivalence relation.

\begin{lemma}\cite[Lemma2.13]{SakoMultidimQW}
If two  discrete-time $d$-dimensional quantum walks are similar,
then there exists a smooth intertwiner between them which is unitary.
\end{lemma}
%
%
%
%

H. Ohno has already used the notion of {\it unitary equivalence} between two quantum walks on countable sets (\cite[Definition 8]{Ohno}).
The condition for unitary equivalence is stronger than similarity.
If two quantum walks $U_1 = (\Hil_1, (U_1^{(t)})_{t}, E_1)$, 
$U_2 = (\Hil_2, (U_2^{(t)})_{t}, E_2)$  are unitary equivalent,
then there exists a unitary operator $T \colon \Hil_1 \to \Hil_2$
which preserves the support of 
vectors, namely
\[\mathrm{supp}(\mu_\xi) = \mathrm{supp}(\mu_{T \xi}), 
\quad \xi \in \Hil_1.\]
It implies that $T$ has finite propagation.
If two quantum walks are unitary equivalent, then they are similar.

\section{Asymptotic behavior of $d$-dimensional smooth quantum walks}
\label{section: asymptotic behavior}

Let $U = (\Hil, (U^t)_{t \in \Z}, E)$ be a $d$-dimensional discrete-time smooth quantum walk.
For a unit vector $\xi$, the coordinate system $\Hil$ on $E$ defines the sequence of the Borel probability measures $\{ \mu_{U^t \xi} \}_{t \in \N}$ on $\R^d$ by the equation
\[\mu_{U^t \xi} (\Omega) = \langle E(\Omega) U^t \xi, U^t \xi \rangle.\]
The dynamical system by $U$ defines the sequence of unit vectors 
$\{U \xi, U^2 \xi, \cdots \}$.
The coordinate system $E$ defines the sequence of probability measures
$\{\mu_{U \xi}, \mu_{U^2 \xi}, \cdots \}$.
In the case that $E$ describes observables of position, 
each probability measure $\mu_{U^t \xi}$ describes the distribution of position of particles whose state is described by $U^t \xi$.

\begin{lemma}\label{lemma: smoothness and integrability}
Suppose that $\xi \in \Hil$ is smooth with respect to $E$.
Then for every $t \in \N$, and every real-valued polynomial function is integrable with respect to $\mu_{U^t \xi}$.
\end{lemma}

\begin{proof}
Let $p$ be an arbitrary real-valued monomial function.
Since $U^t \xi$ is smooth, by Lemma 
\ref{lemma: smoothness for vectors},
$U^t \xi$ is in the domain of $\int_{\R^d} p(\xx) d E(\xx)$.
We have
\begin{eqnarray*}
\infty > 
\left\|
\int_{\R^d} p(\xx) d E(\xx) \cdot U ^t \xi
\right\|^2
= 
\int_{\R^d} p(\xx)^2 \|d E(\xx) U^t \xi\|^2
= 
\int_{\R^d} p(\xx)^2 d \langle E(\xx) U^t \xi , U^t \xi \rangle.
\end{eqnarray*}
It follows that $p^2$ is integrable with respect to $\mu_{U^t \xi}$.
For every polynomial function, its absolute value is dominated by a sum of functions of the form $p^2$.
\end{proof}

It has been already known that for a many quantum walks,
the mean of $\|\xx\|$ with respect to $\mu_{U^t \xi}$ increases linearly
(see e.g.~\cite{KonnoJMSJ}).
We now consider the following sequence $\{ \nu_{U, t, \xi} \}_{t \in \N}$
of probability measures on $\R^d$:
\[\nu_{U, t, \xi} (\Omega) 
= \mu_{U^t \xi} (t\Omega) 
= \langle E(t \Omega) U^t \xi, U^t \xi \rangle.\]
This is the push-forward of the measure $\mu_{U^t \xi}$
along the mapping $\xx \mapsto \xx/t$ on $\R^d$.
In the case that $E$ describes observables of position, 
the measure $\nu_{U, t, \xi}$ is something like the distribution of velocity.
The sequence $\{ \nu_{U, t, \xi} \}_{t \in \N}$ asymptotically has a compact support.
\begin{theorem}\label{theorem: asymptotically having a compact support}
Suppose that a $d$-dimensional discrete-time smooth quantum walk 
$(\Hil, (U^t)_{t \in \Z}, E)$ is smooth.
Then there exists a bounded closed subset $\Omega$ in $\R^d$ such that
for {\rm every} unit vector $\xi$ in $\Hil$,
\[\lim_{t \to \infty} \nu_{U, t, \xi}(\Omega) = 1.\]
\end{theorem}

\begin{proof}
It suffices to show that for every $j \in \{1, \cdots, d\}$, there exists a constant $L_j$ independent of $\xi$ such that
\[\lim_{t \to \infty} \nu_{U, t, \xi}(\R^{j -1} \times [-L_j, L_j] \times \R^{d - j}) = 1.\]
We define a spectral measure $E_j$ on $\R$ by
\[E_j (\Omega_j) = E(\R^{j -1} \times \Omega_j \times \R^{d - j}).\]
It suffices to show the theorem for the $1$-dimensional quantum walk
$(\Hil, (U^t)_{t \in \Z}, E_j)$.
It is easy to show that this quantum walk is also smooth.
In \cite[Proposition 2.28]{SakoIntertwiner}, the above theorem is proved for $1$-dimensional smooth quantum walks and for {\it smooth} unit vectors.

Let us prove it for an {\it arbitrary} unit vector $\xi$.
For a large positive number $a$, define $\eta = \eta(a)$ by 
$E_j([-a, a]) \xi / \| E_j([-a, a]) \xi \|$.
For every positive number $\epsilon$,
there exists a constant $a = a(\xi, \epsilon)$,
such that $\|\eta - \xi\| < \epsilon$. 

On the closed subspace $E_j([-a, a]) \Hil$, the operator $X_j^m = \int_{x \in \R} x^m d E_j(x)$
is equal to
$X_j^m = \int_{-a}^a x^m d E_j(x)$
and therefore bounded.
In particular, $\eta$ is in the domain of $X_j^m$ for every $m \in \N$.
By Lemma \ref{lemma: smoothness for vectors}, $\eta$ is a smooth unit vector.
By \cite[Proposition 2.28]{SakoIntertwiner},
there exists a constant $L_j$ determined by $(\Hil, (U^t)_{t \in \Z}, E_j)$
such that
\[\lim_{t \to \infty} \nu_{U, t, \eta}(\R^{j -1} \times [-L_j, L_j] \times \R^{d - j}) = 1.\]
Since $\|\eta - \xi\| < \epsilon$,
we have
\begin{eqnarray*}
&&
|\nu_{U, t, \xi}(\R^{j -1} \times [-L_j, L_j] \times \R^{d - j})
-
\nu_{U, t, \eta}(\R^{j -1} \times [-L_j, L_j] \times \R^{d - j})|\\
&=&
|\langle E(\R^{j -1} \times [-L_j, L_j] \times \R^{d - j}) \xi, \xi \rangle
-
\langle E(\R^{j -1} \times [-L_j, L_j] \times \R^{d - j}) \eta, \eta \rangle|\\
&<&
2 \epsilon.
\end{eqnarray*}
It follows that
\[\liminf_{t} \nu_{U, t, \xi}(\R^{j -1} \times [-L_j, L_j] \times \R^{d - j})
> 1 - 2 \epsilon.\]
Since $\epsilon$ is arbitrary,
and $L_j$ is independent of $\epsilon$,
we conclude that
\[\lim_{t \to \infty} \nu_{U, t, \xi}(\R^{j -1} \times [-L_j, L_j] \times \R^{d - j}) = 1.\]
\end{proof}

\begin{theorem}\label{theorem: convergence and intertwiner}
Let $U_1 = (\Hil_1, (U_1^t)_{t \in \Z}, E_1)$,
$U_2 = (\Hil_2, (U_2^t)_{t \in \Z}, E_2)$ be 
smooth discrete-time $d$-dimensional quantum walks.
Suppose that $V \colon \Hil_1 \to \Hil_2$ is a uniform unitary intertwiner between them.
For every $\xi \in \Hil_{1}$,
if the sequence 
$\{\nu_{U_1, t, \xi}\}_{t \in \N}$ weakly converges,
then the sequence 
$\{\nu_{U_2, t, V \xi}\}_{t \in \N}$ weakly converges.
\end{theorem}

\begin{proof}
Let $\xi$ be a unit vector in $\Hil_1$.
We compare two sequences of probability measures
$\{\nu_{U_1, t, \xi}\}_{t \in \N}$ and $\{\nu_{U_2, t, V\xi}\}_{t \in \N}$.
For every $\kk$ in $\R^d$,
we have
\begin{eqnarray*}
&&
\left|
\int_{\vv \in \R^d} \exp(\ii \kk \cdot \vv) d\nu_{U_1, t, \xi}(\vv)-
\int_{\vv \in \R^d} \exp(\ii \kk \cdot \vv) d\nu_{U_2, t, V \xi}(\vv)
\right|\\
&=&
\left|
\int_{\xx \in \R^d} \exp(\ii \kk \cdot \xx/t) d\mu_{U_1^t \xi}(\xx)-
\int_{\xx \in \R^d} \exp(\ii \kk \cdot \xx/t) d\mu_{U_2^t V \xi}(\xx)
\right|\\
&=&
\left|
\left\langle
\int_{\xx \in \R^d} \exp(\ii \kk \cdot \xx/t) dE_1(\xx) \cdot U_1^t \xi, U_1^t \xi
\right\rangle
-
\left\langle
\int_{\xx \in \R^d} \exp(\ii \kk \cdot \xx/t) dE_2(\xx) \cdot U_2^t V\xi, U_2^t V\xi
\right\rangle
\right|\\
&=&
\left|
\left\langle
v_1 (\kk/t) \cdot U_1^t \xi, U_1^t \xi
\right\rangle
-
\left\langle
V^{-1} v_2 (\kk/t) V \cdot U_1^t \xi, U_1^t \xi
\right\rangle
\right|
\\
&\le&
\left\|
v_1 (\kk /t)
-
V^{-1} v_2(\kk / t) V
\right\|
=
\left\|
v_1 (\kk /t)
-
V^{-1} \cdot v_2(\kk / t) V  v_1(\kk / t)^{-1} \cdot v_1(\kk / t)
\right\|.
\end{eqnarray*}
Since $V$ is uniform with respect to $E_1$ and $E_2$, as $t$ tends to $\infty$, the operator $v_2(\kk / t) V v_1(\kk / t)^{-1}$ converges to $V$.
It follows that for every $\kk \in \R^d$,
\begin{eqnarray*}
\lim_{t \to \infty} \left|
\int_{\vv \in \R^d} \exp(\ii \kk \cdot \vv) d\nu_{U_1, t, \xi}(\vv)-
\int_{\vv \in \R^d} \exp(\ii \kk \cdot \vv) d\nu_{U_2, t, V \xi}(\vv)
\right|
= 0.
\end{eqnarray*}
If $\{\nu_{U_1, t, \xi}\}_{t \in \N}$ weakly converges,
then for every $\kk$, the sequence
\[\left\{ \int_{\vv \in \R^d} \exp(\ii \kk \cdot \vv) d\nu_{U_2, t, V \xi}(\vv) \right\}_{t \in \N}\] 
converges.
Since the sequence  $\{\nu_{U_2, t, V\xi}\}_{t \in \N}$ asymptotically has a compact support (Theorem \ref{theorem: asymptotically having a compact support}), the sequence weakly converges.
\end{proof}

\section{The category of homogeneous quantum walks}

A $d$-dimensional quantum walk is said to be homogeneous or space-homogeneous, 
if it is compatible with the space structure of $\R^d$.
Let us formulate homogeneous quantum walks, adding a structure of the regular representation to the triplet.

\begin{definition}
Let $t$ be an element of the group of time $\Z$ or $\R$.
The quadruple $U = (\Hil$, $(U^{(t)})_{t}$, $E$, $(\rho(\xx))_{\xx \in \Z^d})$ is called a $d$-dimensional {\rm homogeneous} quantum walk, if it satisfies the following conditions:
\begin{enumerate}
\item
The triplet $(\Hil, (U^{(t)})_{t}, E)$ is a $d$-dimensional quantum walk,
\item
The unitary representation $\rho$ is an action of $\Z^d$ on $\Hil$,
\item
For every Borel subset $\Omega$ of $\R^d$, and every $\xx \in \Z^d$,
\[\rho(\xx)^{-1} E(\Omega) \rho(\xx) = E(\Omega + \xx), \]
\item
For every $t$, and for every $\xx \in \Z^d$, 
$U^{(t)} \rho(\xx) = \rho(\xx) U^{(t)}$,
\item
The dimension of the image of the orthogonal projection $E([0, 1)^d)$
is finite.
\end{enumerate}
The dimension at the last item is called {\rm the degree of freedom}
of $U$.
Define the four kinds of regularity (having finite propagation, analyticity, smoothness, uniformity) of $U$ by those of the triple $(\Hil, (U^{(t)})_{t}, E)$.
In the case that $(U^{(t)})_{t}$ is a representation of $\Z$, we call the quadruple a discrete-time homogeneous quantum walk or simply a homogeneous quantum walk.
In the case that $(U^{(t)})_{t}$ is a strongly continuous representation of $\R$, we call the quadruple a continuous-time homogeneous quantum walk.
\end{definition}

The equation $\rho(\xx)^{-1} E(\Omega) \rho(\xx) = E(\Omega + \xx)$
means that the unitary operator $\rho(\xx)$ shifts the support of vectors $\xi \in \Hil$ in the direction of $- \xx$.
We call $\rho$ a right regular representation of $\Z^d$.
The equation $U \rho(\xx) = \rho(\xx) U$ means that
$U$ is compatible with the regular representation $\rho$.
The subset $[0, 1)^d$ is a Borel fundamental domain of the action of 
$\Z^d$ on $\R^d$.

We can also define morphisms for homogeneous walks, although they are not the subject of this paper.

\begin{definition}
Let $j$ be $1$ or $2$.
Let 
$U_j = \left(\Hil_j, \left(U_j^{(t)} \right)_{t}, E_j, (\rho_j(\xx))_{\xx \in \Z^d} \right)$ be two $d$-dimensional homogeneous quantum walks.
Let $T \colon \Hil_1 \to \Hil_2$ be a bounded linear operator.
The operator $T$ is said to be {\rm a homogeneous intertwiner or a homogeneous morphism} in the category of $d$-dimensional homogeneous quantum walks, if
\begin{itemize}
\item
for every $t$, $T U_1^{(t)} = U_2^{(t)} T$,
\item
for every $\xx \in \Z^d$, $T \rho_1(\xx) = \rho_2(\xx) T$.
\end{itemize}
Define the four kinds of regularity of $T$, using $E_1$ and $E_2$
as in Definition \ref{definition: regularity}.
\end{definition}
%

Now we have two categories. One is the category of general $d$-dimensional walks, the other is that of $d$-dimensional homogeneous walks.
We can easily define a forgetful functor from the category of homogeneous walks to that of general walks,
forgetting the right regular representations.
The homogeneous morphisms map to morphisms in the category of general walks. The converse does not hold true in general.
We will pay attention on general morphisms between two homogeneous walks.

In Section \ref{section: morphisms between 1-dim homogeneous QWs}, we study $1$-dimensional homogeneous analytic quantum walks, and clarify the way to determine the set of uniform morphisms which are not necessarily homogeneous.

\section{General theory of homogeneous quantum walks}
\label{section: general theory of multi-dim QWs}

\begin{theorem}[Subsection 2.6 in \cite{SakoMultidimQW}]\label{theorem: Homogeneous QW in ell-form}
Suppose that $U$ is a $d$-dimensional discrete-time homogeneous quantum walk whose degree of freedom is $n$.
Then $U$ is similar to a homogeneous walk if this form
\[\left( \ell_2(\Z^d) \otimes \C^n, (U^t)_{t \in \Z}, E, (\rho(\xx) \otimes \mathrm{id})_{\xx \in \Z^d} \right),\]
where $E$ is the standard coordinate system of $\ell_2(\Z^d) \otimes \C^n$ introduced in Example \ref{example: standard coordinate system},
and $\rho$ is the right regular representation of $\Z^d$ on $\ell_2(\Z^d)$.
As an intertwiner which give similarity, we can choose 
a unitary which has finite propagation.
Under this choice of unitary, the regularity of the concrete walk is the same as that of the original walk $U$.
\end{theorem}

We apply the inverse Fourier transform $\mathcal{F}^{-1} \colon \ell_2(\Z^d) \otimes \C^n \to L^2(\T_{2 \pi}^d) \otimes \C^n$
to the quadruple.
The space $\T_{2 \pi}^d$ is the $d$-dimensional torus $\R^d / (2 \pi \Z)^d$. For every simply connected domain in $\T_{2 \pi}^d$ has a coordinate $(k_1, \cdots, k_d)$ arising from $\R^d$.
On the Hilbert space $L^2(\T_{2 \pi}^d) \otimes \C^n$, we have $n$-tuple
$\left( \frac{1}{\ii} \frac{\partial}{\partial k_1}, \cdots, \frac{1}{\ii} \frac{\partial}{\partial k_d} \right)$
of mutually commuting self-adjoint operators.
Denote by $\widehat E$ the spectral measure of of the $d$-tuple.
For every $\xx = (x_1, \cdots, x_d) \in \Z^d$, we define $\widehat \rho(\xx)$ by the multiplication operator
by the function $f_{-\xx}(\kk) = \exp(- \ii \xx \cdot \kk)$ on $\T_{2 \pi}^d$

\begin{theorem}[Subsection 2.6 in \cite{SakoMultidimQW}]\label{theorem: Homogeneous QW in ell-form}
Suppose that $U$ is a $d$-dimensional discrete-time homogeneous quantum walk whose degree of freedom is $n$.
Then $U$ is similar to a homogeneous walk if this form
\[\left( L^2(\T_{2 \pi}^d) \otimes \C^n, \left( \widehat U^t \right)_{t \in \Z}, \widehat E, (\widehat \rho(\xx) \otimes \mathrm{id})_{\xx \in \Z^d} \right),\]
where
every entry of $\widehat{U}$ is a multiplication operator of a function on
$L^2(\T_{2 \pi}^d)$.
As an intertwiner which give similarity, we can choose 
a unitary which has finite propagation.
Under this choice of unitary, 
the regularity of $U$ is rephrased as follows:
\begin{enumerate}
\item
$U$ has finite propagation, if and only if
every entry of $\widehat{U}$ is a multiplication operator of a function which is a finite linear combination of $\{f_{-\xx}(\kk) = \exp(- \ii \xx \cdot \kk) \}_{\xx \in \Z^d}$.
\item
$U$ is analytic (smooth, or uniform), if and only if
every entry of $\widehat{U}$ is a multiplication operator of an analytic (a smooth, or a continuous, respectively) function on $\T_{2 \pi}^d$.
\end{enumerate}
\end{theorem}

In our new category of homogeneous quantum walks, an inverse Fourier transform of a quantum walk is not merely an intermediate product of study.
It is also an object of the category.
The author thinks that this is merit of our new framework.

\begin{theorem}[Theorem 3.4 in \cite{SakoMultidimQW}]\label{theorem: limit of HQW}
Every $d$-dimensional homogeneous discrete-time analytic quantum walk $U$ has a weak limit distribution $\lim_{t \to \infty} \nu_{U, t, \xi}$ of velocity for every initial unit vector $\xi$.
The support of the weak limit is compact. 
\end{theorem}

\begin{remark}
Limit distributions of homogeneous quantum walks has attracted much attention of researchers. 
A preceding paper insists the above theorem by the following argument:
``For every homogeneous quantum walk,
we can diagonalize its inverse Fourier transform.
The eigenvalue functions are smooth.
Therefore, the walk has limit distribution.''
Unfortunately, this argument is inappropriate.

The crucial problem is that some papers did not clarify what they call quantum walks, so the theorem was not clarified.
Moreover, some preceding papers did not pay attention on the following types of examples:
\begin{itemize}
\item
The $2$-dimensional homogeneous quantum walk 
$U
=
\dfrac{1}{2}
\left(
\begin{array}{cc}
S_{\rm r} + S_{\rm u} & - S_{\rm r}^{-1} + S_{\rm u}^{-1} \\ 
S_{\rm r} - S_{\rm u} & S_{\rm r}^{-1} + S_{\rm u}^{-1}
\end{array}
\right)
$
does not have a differentiable eigenvalue function,
although it has finite propagation.
See \cite[Section 3]{SakoMultidimQW}.
\item
For the $1$-dimensional homogeneous quantum walk 
$\left(
\begin{array}{cc}
a S & -b S\\
b & a
\end{array}
\right)$
with finite propagation, the eigenvalue function is multi-valued.
See \cite[Subsection 9.2]{SaigoSako}.
\end{itemize}
Therefore, a naive use of diagonalization is not appropriate for the proof.
\end{remark}

We can generalize Theorem \ref{theorem: limit of HQW} for time-periodic homogeneous analytic quantum walks.
Let $V_1, \cdots, V_p$ be analytic unitary operators on the Hilbert space $(\Hil, E)$ with a $d$-dimensional coordinate system and suppose that each unitary operator commutes with the right regular representation $\rho(\xx)$ of $\Z^d$.
For a natural number $t$, define $V_n$ by the following:
if $t \equiv j \ (\mod p)$, then $V_t = V_j$.
For $t \in \N$, define $U^{(t)}$ by $U^{(0)} = 1, U^{(t)} = V_t V_{t -1} \cdots V_2 V_1$.
Let us call the quadruple $(\Hil, (U^{(t)})_{t \in \N}, E, \rho)$ a 
$d$-dimensional space-homogeneous {\it time-periodic} discrete-time analytic quantum walk.

\begin{theorem}\label{theorem: convergence of time-periodic homogeneous walk}
Every $d$-dimensional space-homogeneous {\rm time-periodic} discrete-time analytic quantum walk $\left( \Hil, \left( U^{(t)} \right)_{t \in \N}, E, \rho \right)$ has a weak limit distribution of velocity for {\rm every} initial unit vector.
The support of the weak limit is compact. 
\end{theorem}

\begin{proof}
For every $t \in \N$, there exist unique $q = q(t) \in \Z_{\ge 0}$ and 
$r = r(t) \in \{0, 1, \cdots, d - 1\}$ such that $t = p q + r$.
Define $W^{(r)}$ by
\[
W^{(0)} = 1, \quad
W^{(1)} = V_1, \quad
W^{(2)} = V_1 V_2, \quad
\cdots, \quad
W^{(d -1)} = V_1 V_2 \cdots V_{d - 1}.\]
We have $U^{(t)} = W^{(r)} {U^{(p)}}^q$.
Note that the quadruple $(\Hil, (U^{(p q)})_{p \in \N}, E, \rho)$ is a usual homogeneous quantum walk, and has a weak limit for every initial unit vector $\xi$ (Theorem \ref{theorem: limit of HQW}).

Choose and fix a unit vector $\xi$ of $\Hil$.
The $t$-th distribution $\nu_{U, t, \xi}$ of the velocity is defined by the following:
\[\nu_{U, t, \xi}(\Omega) 
= \left\langle E(t \Omega) U^{(t)} \xi, U^{(t)} \xi \right\rangle
= \left\langle E(t \Omega) W^{(r)} {U^{(p)}}^q \xi, W^{(r)} {U^{(p)}}^q \xi \right\rangle.
\]
Let us denote by $\phi_{U, t, \xi}(\kk)$ the characteristic function 
$\kk \mapsto \int_{\xx \in \R^d} \exp(\ii \xx \cdot \kk) d \nu_{U, t, \xi}$ 
of the probability measure $\nu_{U, t, \xi}$.

For every wavenumber $\kk \in \R^d$, the value of the characteristic function $\phi_{U, t, \xi}(\kk)$ of $\nu_{U, t, \xi}$ is equal to the following:
\begin{eqnarray*}
\int_{\xx \in \R^d} \exp(\ii \xx \cdot \kk) d \nu_{U, t, \xi}
&=&
\int_{\xx \in \R^d} \exp \left( \ii \frac{\xx}{t} \cdot \kk \right) 
d \mu_{U^{(t)} \xi}(\xx)\\
&=&
\left\langle
\int_{\xx \in \R^d} \exp \left( \ii \frac{\xx}{t} \cdot \kk \right) 
d E(\xx) \cdot W^{(r)} {U^{(p)}}^q \xi, W^{(r)} {U^{(p)}}^q \xi
\right\rangle\\
&=&
\left\langle
\left( {W^{(r)}}^{-1} \cdot
\int_{\xx \in \R^d} \exp \left( \ii \frac{\xx}{t} \cdot \kk \right)
d E(\xx) \cdot W^{(r)} \right) {U^{(p)}}^q \xi, {U^{(p)}}^q \xi
\right\rangle.
\end{eqnarray*}
Since $r$ is an element of a finite subset, and $W^{(r)}$ is uniform,
the commutator of $W^{(r)}$ and $v_E(\kk / t) = \int_{\xx \in \R^d} 
\exp \left( \ii \xx / t \cdot \kk \right) 
d E(\xx)$
converges to $0$ as $t$ tends to $\infty$.
Therefore, if $t$ is large, then the above quantity is almost the same as the following
\begin{eqnarray*}
&&\left\langle
\int_{\xx \in \R^d} \exp \left( \ii \frac{\xx}{t} \cdot \kk \right) 
d E(\xx) \cdot {U^{(p)}}^q \xi, {U^{(p)}}^q \xi
\right\rangle\\
&=&
\left\langle
\int_{\xx \in \R^d} \exp \left( \ii \frac{\xx}{p q} \cdot \frac{ p q}{t}  \kk \right) 
d E(\xx) \cdot {U^{(p q)}} \xi, {U^{(p q)}} \xi
\right\rangle\\
&=&
\int_{\xx \in \R^d} \exp \left( \ii \frac{\xx}{p q} \cdot \frac{ p q}{t}  \kk \right) 
d \mu_{U^{(p q)} \xi}(\xx)
=
\int_{\xx \in \R^d} \exp \left( \ii \xx \cdot \frac{p q}{t}  \kk \right) 
d \nu_{U, p q, \xi}.
= \phi_{U, p q, \xi} \left( \frac{p q}{t}  \kk \right).
\end{eqnarray*}
It follows that for every positive number $\epsilon$, and $\kk \in \R^d$,  there exists a natural number $N_1 = N_1(\epsilon, \kk)$ such that
for every $t$ at least $N_1$,
\[\left| \phi_{U, t, \xi}(\kk) - \phi_{U, p q, \xi} \left( \frac{p q}{t}  \kk \right) \right| < \epsilon.\]

Since the sequence of the probability measures $\{\nu_{U, pq, \xi}\}_{q \in \N}$ asymptotically has compact support, 
the characteristic functions $\left\{ \phi_{U, pq, \xi}(\kk) = \int_{\xx \in \R^d} \exp(\ii \xx \cdot \kk) d \nu_{U, pq, \xi} \right\}_{q \in \N}$ of the sequence are equicontinuous.
It follows that 
for every positive number $\epsilon$, and $\kk \in \R^d$,  there exists a natural number $N_2 = N_2(\epsilon, \kk)$ such that
for every $q$ at least $N_2$,
\[\left| \phi_{U, p q, \xi} \left( \frac{p q}{t}  \kk \right) 
- \phi_{U, p q, \xi} \left( \kk \right)\right| < \epsilon.\]

Since the quantum walk $(\Hil, (U^{(p q)})_{q \in \N}, E, \rho)$ is a usual homogeneous analytic quantum walk, it has a weak limit distribution.
It follows that for every $\kk \in \R^d$, as $q$ tends to $\infty$, the sequence $\{\phi_{U, p q, \xi} \left( \kk \right)\}_{q \in \N}$ converges.
Combining the above two inequalities, we obtain that for every $\kk$, the sequence $\{\phi_{U, t, \xi}(\kk)\}_{t \in \N}$ converges.

Since the measures $\{\nu_{U, t, \xi}\}_{t \in \N}$ asymptotically has compact support, the point-wise convergence of the characteristic functions implies that the sequence $\{\nu_{U, t, \xi}\}_{t \in \N}$ weakly converges.
\end{proof}

\section{Structure of $1$-dimensional homogeneous quantum walks}
\label{section: structure theorems of $1$-dim HQW}

Our next project is to study $1$-dimensional homogeneous {\it analytic} quantum walks $U = (\Hil, (U^t)_{t \in \Z}, E, (\rho(x))_{x \in \Z})$.
We finally get a way to determine how many uniform intertwiners two such quantum walks have.
We would like to study not only homogeneous intertwiners but also non-homogeneous intertwiners.
For this purpose, forgetting the fourth entry $\rho$, we study $1$-dimensional homogeneous quantum walks in the cagegory of $1$-dimensional quantum walks.

Suppose that $\wU$ is a $1$-dimensional quantum walk of the form
\[\widehat U = 
\left( L^2(\T_{2 \pi}) \otimes \C^n, \left( \widehat U^t \right)_{t \in \Z}, \widehat E \right),\]
where
\begin{itemize}
\item
the space $\T_{2 \pi}$ is the quotient space $\R / 2 \pi \Z$ and every open interval in $\T_{2 \pi}$ has a coordinate $k$ arising form $\R$,
\item
$L^2(\T_{2 \pi}) \otimes \C^n$ can be regarded as the set of all the column vectors consisting of $n$ vectors in $L^2 (\T_{2 \pi})$,
\item
the operator $\widehat {U}$ can be expressed as a $(n \times n)$-matrix whose entries are multiplication operators on $L^2 (\T_{2 \pi})$,
\item
$\widehat E$ is the spectral decomposition of the self-adjoint operator
$\frac{1}{\ii} \frac{d}{dk} \otimes \mathrm{id}$.
\end{itemize}
It has been already shown in Theorem \ref{theorem: Homogeneous QW in ell-form} that every $1$-dimensional homogeneous quantum walk $U$
is similar to a walk of this form.
In such a case, we call $\widehat U$ an inverse Fourier transform of $U$.

For the rest of this paper, we focus on $1$-dimensional analytic homogeneous quantum walks of the form $\widehat U$.
Because each matrix entries of $\wU$ gives an analytic function, for every $k \in \R$ we obtain a $(n \times n)$-complex matrix $\wU(k)$.
Thus $\wU(\ \cdot\ )$ gives an analytic map from a domain containing $\R$ to the space $M_n(\C)$ of the complex matrices.
The positive number $2 \pi$ is a period of the analytic map $\wU (\ \cdot \ )$.

Let us keep in mind that $\wU$ is not an intermediate product of our study.
The triple $\widehat U = 
\left( L^2(\T_{2 \pi}) \otimes \C^n, \left( \widehat U^t \right)_{t \in \Z}, \widehat E \right)$ is an object in our new category.
This strategy makes our study simple.

\subsection{Eigenvalue functions of $1$-dimensional homogeneous quantum walks}

Denote by $\T$ the set of complex numbers whose absolute values are $1$.
\begin{definition}
An analytic function $\lambda \colon \R \to \T$ is called {\rm an eigenvalue function} of 
\[\widehat U = 
\left( L^2(\T_{2 \pi}) \otimes \C^n, \left( \widehat U^t \right)_{t \in \Z}, \widehat E \right)\]
(or an eigenvalue function of $U$), if for every $k \in \R$,
$\lambda(k)$ is an eigenvalue of $\wU(k)$.
\end{definition}

Since the map $\wU(k)$ has a period, and since for every $k$ the set of eigenvalues of $\wU(k)$ is a finite subset of $\T$,
$\lambda$ has to have period.
There exists $d \in \N$ such that $2 \pi d$ is a period of $\lambda$.

\begin{example}
Although $2 \pi$ is a period of the map $\wU(k)$,
$2 \pi$ is not necessarily a period of the eigenvalue function $\lambda$.
Let us consider the walk
$U = 
\left(
\begin{array}{cc}
a S & -b S\\
b & a
\end{array}
\right)$
on $\ell_2(\Z) \otimes \C^2$.
Its inverse Fourier transform $\wU$ is 
$\left(
\begin{array}{cc}
a \exp(\ii k) & -b \exp(\ii k)\\
b & a
\end{array}
\right)$.
The analytic function 
\[\lambda(k) 
= \exp (\ii k /2) \left( a \cos (k / 2)
- \ii \sqrt{1 - a^2 \cos^2 (k /2)} \right)\]
is an eigenvalue function of $\wU$.
The minimal period of $\lambda$ is $4 \pi$.
\end{example}

\begin{example}
The eigenvalue function $\lambda$ has to have a period of the form $2 \pi d$, but
the minimal period is not necessarily of the form $2 \pi d$.

Let us consider the walk
$U = 
\left(
\begin{array}{ccc}
0 & S & 0\\
0 & 0 & S\\
1 & 0 & 0
\end{array}
\right)$
on $\ell_2(\Z) \otimes \C^3$.
Its inverse Fourier transform $\wU$ is 
$\left(
\begin{array}{ccc}
0 & \exp(\ii k) & 0\\
0 & 0 & \exp(\ii k)\\
1 & 0 & 0
\end{array}
\right)$.
The analytic function $\lambda(k) = \exp (2 \ii k /3)$
is an eigenvalue function of $\wU$.
The number $6 \pi = 2 \pi \cdot 3$ is a period of $\lambda$. 
The minimal period of $\lambda$ is $3 \pi$.
\end{example}

The author would like to emphasize that existence of an eigenvalue function is non-trivial fact,
since there might exists $k \in \R$ such that $\wU(k)$ has a multiple eigenvalue.

\begin{theorem}[Proposition 4.8 in \cite{SaigoSako}]
For every $1$-dimensional analytic homogeneous quantum walk of the form 
\[\left( L^2(\T_{2 \pi}) \otimes \C^n, \left( \widehat U^t \right)_{t \in \Z}, \widehat E \right),\] 
there exists a finite family of eigenvalue functions $\lambda_1, \cdots, \lambda_m \colon \R \to \C$ of $\wU$
satisfying that for every $k \in \R$
the set of all the eigenvalues of $\wU(k)$ is equal to
\[\{\lambda_i(k + 2 \pi j) \ | \ i \in \{1, \cdots, m\}, j \in \Z\}.\]
\end{theorem}

\begin{example}
Let $U$ be the $3$-state Grover walk.
The inverse Fourier transform is given by 
\[\wU = 
\dfrac{1}{3}
\left(
\begin{array}{ccc}
\exp(- \ii k) &    0    & 0       \\
0       &    1     & 0       \\
0       &    0     & \exp(\ii k)
\end{array}
\right)
\left(
\begin{array}{ccccc}
1 & -2 & -2 \\
-2 & 1 & -2 \\
-2 & -2 & 1 \\
\end{array}
\right).
\]
Define $\lambda_1, \lambda_2 \colon \R \to \T$
by
\[\lambda_1(k) = -1, \quad \lambda_2 = \frac{2 + \cos k}{3} + \frac{1}{3} \ii \sin \frac{k}{2} \sqrt{10 + 2 \cos k}.\]
The set of all the eigenvalues of $\wU(k)$ is equal to $\{\lambda_1(k), \lambda_2(k), \lambda_2(k  + 2 \pi)\}$.
\end{example}

\subsection{Model quantum walks and the structure theorem}

Every $1$-dimensional homogeneous quantum walks is decomposed as a direct sum of model quantum walks.
Let us define model quantum walks.
Let $\lambda \colon \R \to \C$ is an analytic function and let $p$ be a positive number and a period of $\lambda$. The function $\lambda$ can be a constant function.
Let $\T_{p}$ be the torus $\R / p \Z$ of length $p$.
Note that every interval in $\T_{p}$ has a coordinate $k$ arising from $\R$.
Also note that the self-adjoint operator $\frac{1}{\ii} \frac{d}{d k}$ gives a spectral measure on $\frac{2 \pi}{p} \Z \subset \R$.
Denote by $M[\lambda]$ be the multiplication operator on $L^2(\T_p)$
by $\lambda$.

\begin{definition}
Define {\rm a model quantum walk} $\wU_{\lambda, p}$
by the triple
$\left( L^2(\T_p), (M[\lambda]^t)_{t \in \Z}, \frac{1}{\ii} \frac{d}{d k} \right)$.
\end{definition}
All the model quantum walk is analytic.
Note that the spectrum of the unitary operator $M[\lambda]$
is the image of $\lambda$.

What happens when $p$ is not the minimal period of $\lambda$?
\begin{lemma}
[Proposition 3.3 in \cite{SakoIntertwiner}]
\label{lemma: decomposition of model QW}
Let $n$ be a natural number
There exists a unitary intertwiner with finite propagation between
$\wU_{\lambda, n p}$ and the direct sum $\left(\wU_{\lambda, p}\right)^{\oplus n}$.
\end{lemma}

\begin{proof}
We shall give a self-contained proof.
The quantum walk $\wU_{\lambda, n p}$ acts on $L^2(\T_{np})$,
and the quantum walk $\left(\wU_{\lambda, p}\right)^{\oplus n}$
acts on $L^2(\T_p)^{\oplus n}$.
For $x \in \frac{2\pi}{n p} \Z$, define $f_x \in L^2(\T_{np})$ by
$f_x(k) = \exp(\ii x k)$.
For $y \in \frac{2\pi}{p} \Z$, define $g_x \in L^2(\T_{p})$ by
$g_y(l) = \exp(\ii y l)$.
For every $x \in \frac{2\pi}{n p} \Z$, there exists a unique pair $(y(x), r(x)) \in  \left( \frac{2\pi}{p} \Z, \{0, 1, 2, \cdots, n -1\} \right)$
such that $x = y(x) + \frac{2 \pi}{n p} r(x)$.

Define a unitary operator $V \colon L^2(\T_{np}) \to L^2(\T_p) \otimes \ell_2(\{0, 1, \cdots, n-1\})$
by
\[f_x \mapsto (0, \cdots, 0, g_{y(x)}, 0, \cdots, 0),\]
where $g_{y(x)}$ is put on the $r(x)$-th entry.
The function $\lambda$ on $\T_p$ is of the form
$\sum_{y \in  \frac{2\pi}{p} \Z} c_y g_y$.
The function $\lambda$ on $\T_{n p}$ is of the form
$\sum_{y \in  \frac{2\pi}{p} \Z} c_y f_y$.
For every $x \in \frac{2\pi}{n p} \Z$,
we have
\begin{eqnarray*}
V M[\lambda] f_x &=& V \sum_{y \in  \frac{2\pi}{p} \Z} c_y f_{x + y}\\
&=&
\left( 
0, \cdots, 0, \sum_{y \in  \frac{2\pi}{p} \Z} c_y g_{y(x) + y}, 0, \cdots, 0
\right)
\\
&=&
(0, \cdots, 0, M[\lambda] g_{y(x)}, 0, \cdots, 0).
\end{eqnarray*}
It follows that $V$ intertwines $\wU_{\lambda, n p}$ and the direct sum $\left(\wU_{\lambda, p}\right)^{\oplus n}$.
By definition, $V$ has finite propagation.
\end{proof}

\begin{theorem}[Proposition 5.6 in \cite{SaigoSako}]\label{theorem: structure theorem}
Let $U$ be an arbitrary $1$-dimensional homogeneous analytic quantum walk.
Suppose that
eigenvalue functions $\lambda_1, \cdots, \lambda_m \colon \R \to \T$ of $\wU$ and elements $p(1), \cdots, p(m)$ of $2 \pi \N$ satisfy the following conditions:
\begin{enumerate}
\item
For every $j \in \{1, \cdots, m\}$, $p(j)$ is one of the periods of 
$\lambda_j$.
\item
For every $k \in \R$, the list of real numbers
\begin{itemize}
\item
$\lambda_1(k)$, $\lambda_1(k + 2 \pi)$, $\cdots$, $\lambda_1(k + p(1) - 2 \pi)$, 
\item
$\lambda_2(k)$, $\lambda_2(k + 2 \pi)$, $\cdots$, $\lambda_2(k + p(2) - 2 \pi)$,
\item
$\quad \vdots \quad $, 
\item
$\lambda_m(k)$, $\lambda_m(k + 2 \pi)$, $\cdots$, $\lambda_m(k + p(m) - 2 \pi)$.
\end{itemize}
coincides with the list of eigenvalues of $\wU(k)$ including multiplicity.
\end{enumerate}
Then there exists an analytic unitary which intertwines $\wU$ and the direct sum $\oplus_{j = 1}^m \wU_{\lambda_j, p(j)}$ of the model quantum walks.
\end{theorem}

The most non-trivial part of the proof is a construction of analytic unitary intertwiner between $\wU$ and the direct sum of model quantum walks
$\oplus_{j = 1}^m \wU_{\lambda_j, p(j)}$.
This is done when we find global analytic sections of unit eigenvectors of $\wU(k)$.
The precise argument is in Proposition 4.11 of \cite{SaigoSako}.
The analytic unitary intertwiner does not necessarily have finite propagation.

In the case that $p(j)$ is not the minimal period, we may decompose the model quantum walk $\wU_{\lambda_j, p(j)}$ to a direct sum of model quantum walks.
Thus we obtain direct sum of constant quantum walks and model quantum walks with minimal periods. 

\section{Morphisms between $1$-dimensional homogeneous quantum walks}
\label{section: morphisms between 1-dim homogeneous QWs}

To determine the set of the uniform intertwiners between given $1$-dimensional homogeneous analytic quantum walks,
we have only to determine the set of the uniform intertwiners between
model quantum walks, by Theorem \ref{theorem: structure theorem}.
Note that there exists no non-zero intertwiner between constant a quantum walk and a model quantum walk of non-constant analytic function, since
the latter unitary operator has no eigenvector.
By Lemma \ref{lemma: decomposition of model QW}, we may focus on model quantum walks $\wU_{\lambda, p}$ defined by non-constant analytic functions $\lambda$ and its minimal periods $p$,

\subsection{The set of morphisms between two model QWs associated to minimal periods}

\begin{theorem}
\label{theorem: existence of an intertwiner}
Let $j$ be $1$ or $2$.
Let $\lambda_j \colon \R \to \T$ be a non-constant analytic function.
Let $p(j)$ be the minimal period of $\lambda_j$.
The the following conditions are equivalent:
\begin{enumerate}
\item
There exists a real number $l$ such that for every real number $k$,
$\lambda_2(k) = \lambda_1(k - l)$.
\item
There exists a unitary intertwiner from $\wU_{\lambda_1, p(1)}$
to $\wU_{\lambda_2, p(2)}$ with finite propagation.
\item
These walks $\wU_{\lambda_1, p(1)}$, $\wU_{\lambda_2, p(2)}$ are similar.
\item
There exists a non-zero uniform intertwiner from $\wU_{\lambda_1, p(1)}$
to $\wU_{\lambda_2, p(2)}$.
\end{enumerate}
\end{theorem}

\begin{proof}
The implications $(2) \Rightarrow (3)$ and $(3) \Rightarrow (4)$ are trivial.
Suppose the condition $(1)$.
In this case $p := p(1) = p(2)$. Define $\sigma_l \colon L^2(\T_p) \to L^2(\T_p)$ by $\sigma_l \xi(k) = \xi(k - l)$.
It is straightforward to show that $\sigma_l$ is an intertwiner from $M[\lambda_1]$ to $M[\lambda_2]$.
The unitary operator $\sigma_l$ maps an eigenvector of $\frac{1}{\ii}\frac{d}{d k}$ to an eigenvector of $\frac{1}{\ii}\frac{d}{d k}$,
and the eigenvalue does not change.
It follows that $\sigma_l$ has finite propagation.

The implication $(4) \Rightarrow (1)$ is non-trivial.
It has been shown in \cite[Proposition 3.10]{SakoIntertwiner}.
\end{proof}

Recall that uniformity is the weakest regularity for intertwiners.
The author does not think that a non-uniform intertwiner is related to the theory of quantum walks.

\begin{example}
Let $a_1$ and $a_2$ be real numbers satisfying $0 < a_j < 1$, $a_1 \neq a_2$.
Define $b_j$ by $\sqrt{1 - a_j^2}$.
Then there exists no non-zero uniform intertwiner between
\[
U_1 = 
\left(
\begin{array}{cc}
a_1 S^{-1} & -b_1 S^{-1}\\
b_1 S & a_1 S
\end{array}
\right),
\quad
U_2 = 
\left(
\begin{array}{cc}
a_2 S^{-1} & -b_2 S^{-1}\\
b_2 S & a_2 S
\end{array}
\right).
\]
Indeed, the eigenvalue functions of the former walk are
\[\lambda_{1, \pm}(k) = a_1 \cos k \pm \ii \sqrt{1 - a_1^2 \cos^2 k},\]
and the eigenvalue functions of the latter walk are
\[\lambda_{2, \pm}(k) = a_2 \cos k \pm \ii \sqrt{1 - a_2^2 \cos^2 k}.\]
The positive number $2 \pi$ is the minimal period of all of them.
The walk $U_j$ is similar to the direct sum 
$\wU_{\lambda_{j, +}} \oplus \wU_{\lambda_{j, +}}$.
The function $\lambda_{1, +}$ is not a translate of $\lambda_{2, +}$ 
nor $\lambda_{j, -}$, and 
$\lambda_{1, -}$ is not a translate of $\lambda_{2, +}$ 
nor $\lambda_{j, -}$.
Now, we can apply the above theorem.
There is no non-zero uniform intertwiner between 
$\wU_{\lambda_{1, +}} \oplus \wU_{\lambda_{1, +}}$ and 
$\wU_{\lambda_{2, +}} \oplus \wU_{\lambda_{2, +}}$.
Therefore,
There is no non-zero uniform intertwiner between $U_1$ and $U_2$.
\qed
\end{example}

We can determine the set of uniform intertwiners between two model quantum walks as follows.

\begin{theorem}[Proposition 3.10 in \cite{SakoIntertwiner}]
\label{theorem: the set of intertwiners}
Let $\lambda_1, \lambda_2 \colon \R \to \T$ be non-constant analytic functions.
Let $l$ be a real number.
Suppose that $\lambda_2(k) = \lambda_1(k - 1)$.
Denote by $p$ the minimal period of $\lambda_1$ and $\lambda_2$.
Denote by $\sigma_l \colon L^2(\T_p) \to L^2(\T_p)$ the unitary operator given by the translation by $l$.
Denote by $C(\T_p)$ the set of all the complex valued continuous functions on $\T_p$.
Then the set of all the uniform intertwiners 
from $\wU_{\lambda_1, p}$
to $\wU_{\lambda_2, p}$
is equal to
\[\left\{ M[f] \sigma_l \ |\ f \in C(\T_p) \right\}.\]

The intertwiner $M[f] \sigma_l$ is smooth, if and only if $f$ is smooth.
The intertwiner $M[f] \sigma_l$ is analytic, if and only if $f$ is analytic.
\end{theorem}

\subsection{How to determine the set of morphisms between two $1$-dim homogeneous QWs}

Following below, we can determine whether there exists a non-zero uniform intertwiner 
between two $1$-dimensional homogeneous analytic quantum walks $U_1$, $U_2$:
\begin{enumerate}
\item
Calculate the inverse Fourier transforms $\wU_j(k)$,
\item
Calculate the characteristic polynomials of $\wU_j(k)$,
\item
Calculate the lists of eigenvalues of $\wU_j(k)$,
\item
Express the lists by analytic functions.
\end{enumerate}
Here we obtain the direct sum decomposition of $\wU$ into model quantum walks
$\oplus_{j = 1}^m \wU_{\lambda_j, p(j)}$.
In the case that $\lambda_j$ is not a constant function,
we further decompose the model quantum walk $\wU_{\lambda_j, p(j)}$
into the direct sum of the model quantum walk with minimal period.
\begin{enumerate}
\setcounter{enumi}{4}
\item
Compare two lists. If we can find an eigenvalue function $\lambda_1 \colon \R \to \T$ of $\wU_1$ and an eigenvalue function $\lambda_2 \colon \R \to \T$ of $\wU_2$ satisfying that $\lambda_2$ is a translate of $\lambda_1$, then we get to know there exists a non-zero uniform intertwiner between them.
Otherwise, there exists no non-zero uniform intertwiner between them.
\end{enumerate}
By the above procedure, we can also determine whether there exists a non-zero smooth (or analytic) intertwiner.

\begin{example}
There exists an analytic isometric intertwiner from the quantum walk
\[U_2 = \dfrac{1}{2}\left(
\begin{array}{cc}
S^3 + S & S - S^{-1}\\
S - S^{-1} & S^{-1} + S^{-3}
\end{array}
\right)\]
to the $1$-dimensional $4$-state Grover walk $U_4$ defined in Example \ref{example: 4-state Grover walk}.
Indeed,
the eigenvalue functions of $U_2$ are 
\[\lambda_\pm(k) = \dfrac{\cos k + \cos 3k}{2} \pm \ii \sin k \sqrt{1 + 4 \cos^4 k}.\]
Therefore, there exists an analytic unitary intertwiner between $U_2$
and the direct sum $\wU_{\lambda_+, 2 \pi} \oplus \wU_{\lambda_-, 2 \pi}$
of model quantum walks.
The eigenvalue functions of $U_4$ are 
\[1, - 1, \lambda_\pm(k).\]
Therefore, there exists an analytic unitary intertwiner between $U_4$
and the direct sum $1 \oplus (-1) \oplus \wU_{\lambda_+, 2 \pi} \oplus \wU_{\lambda_-, 2 \pi}$ acting on $L^2(\T_{2\pi})^{\oplus 4}$.
Since there exists a (trivial) isometric and analytic intertwiner from $\wU_{\lambda_+, 2 \pi} \oplus \wU_{\lambda_-, 2 \pi}$ to $1 \oplus (-1) \oplus \wU_{\lambda_+, 2 \pi} \oplus \wU_{\lambda_-, 2 \pi}$, 
there exists an isometric and analytic intertwiner from $U_2$ to $U_4$.
\end{example}

\begin{corollary}
Let $U_1$ and $U_2$ be $1$-dimensional homogeneous analytic quantum walks.
Suppose that there exists a non-zero uniform intertwiner from $U_1$ to $U_2$. Then there exists a non-zero analytic intertwiner from $U_1$ to $U_2$, which is a partial isometry.
\end{corollary}

\begin{proof}
In this case, there exist a direct summand $\wU_{\lambda_1, p(1)}$ of $U_1$ and a direct summand $\wU_{\lambda_2, p(2)}$ of $U_2$
such that there exists a non-zero uniform intertwiner between them.
By theorem \ref{theorem: existence of an intertwiner}, there exists 
an analytic unitary intertwiner from $\wU_{\lambda_1, p(1)}$ to $\wU_{\lambda_2, p(2)}$.
\end{proof}

\subsection{Indecomposable quantum walks and some factorization results}

\begin{corollary}\label{corollary: indecomposable model quantum walk}
Suppose that there exists a unifrom unitary intertwiner from $U$ to
a model quantum walk defined by the minimal period.
Then there is no uniform unitary intertwiner from $U$ to a non-trivial direct sum of two $1$-dimensional quantum walks.
\end{corollary}

\begin{proof}
We may assume that the quantum walk is the model quantum walk 
$\wU_{\lambda, p}$.
By Theorem \ref{theorem: the set of intertwiners},
the set of uniform intertwiners from $\wU_{\lambda, p}$
to $\wU_{\lambda, p}$
is equal to
\[\left\{ M[f] \ |\ f \in C(\T_p) \right\}.\]
The space does not contain a projection other than $0, 1$.
Suppose that there were a direct sum decomposition.
Then there would exist a non-trivial uniform projection which commutes with $\wU_{\lambda_1, p}$.
\end{proof}

\begin{definition}
If a walk $(\Hil, (U^t)_{t \in \Z}, E)$ is similar to a non-trivial direct sum of two quantum walks, $U$ is said to be {\rm decomposable}.
Otherwise, it is said to be {\rm indecomposable}.
\end{definition}

\begin{corollary}
For every $1$-dimensional homogeneous analytic quantum walk $U$,
the following conditions are equivalent:
\begin{enumerate}
\item
$U$ is indecomposable.
\item
There exists an analytic unitary intertwiner between $U$ and some model quantum walk defined by the minimal period.
\item
The walk $U$ is similar to some model quantum walk defined by the minimal period.
\item
There exists a uniform invertible intertwiner between $U$ and some model quantum walk defined by the minimal period.
\end{enumerate}
\end{corollary}

\begin{proof}
The implications $(2) \Rightarrow (3)$ and $(3) \Rightarrow (4)$ are trivial.

Suppose that $(4)$ holds. Let $V$ be the invertible intertwiner from $U$ to the model quantum walk defined by the minimal period.
Let $W |V|$ be the polar decomposition of $V$.
Then $W$ a uniform unitary intertwiner between $U$ and the model walk.
By Corollary \ref{corollary: indecomposable model quantum walk}, we obtain $(1)$.

Suppose that $U$ is indecomposable.
By Theorem \ref{theorem: structure theorem}, and by indecomposability of $U$, there exists an analytic unitary which intertwines $U$ and a model quantum walk $\wU_{\lambda, p}$. By Lemma \ref{lemma: decomposition of model QW}
$p$ is the minimal period of $\lambda$.

\end{proof}

\begin{example}
Let $a, b$ be real numbers satisfying $0 < a < 0$, $a^2 + b^2 = 1$.
The walk
$U_1 = \left(
\begin{array}{cc}
a S^{-1} & - b S^{-1}\\
b S  & a S
\end{array}
\right)$
is decomposable.
Indeed, the eigenvalue functions are
\[a \cos k \pm \ii \sqrt{1 - a^2 \cos^2 k}. \]
It follows that $U_1$ is similar to a direct sum of two model quantum walks.

The walk $U = \left(
\begin{array}{cc}
a S & -b S\\
b & a
\end{array}
\right)$ is indecomposable.
Indeed,
the eigenvalue function is
\[\lambda(k) 
= \exp (\ii k /2) \left( a \cos (k / 2)
- \ii \sqrt{1 - a^2 \cos^2 (k /2)} \right)\]
The minimal period is $4 \pi$.
The walk is similar to the model quantum walk $\wU_{\lambda, 4 \pi}$
\end{example}

The following corollaries are factorization results on $1$-dimensional homogeneous walks.

\begin{corollary}
Every $1$-dimensional homogeneous analytic quantum walk $U$
is similar to a direct sum of
\begin{itemize}
\item
constant quantum walks,
\item
and indecomposable model quantum walks.
\end{itemize}
One of two items can be empty.
\end{corollary}

\begin{corollary}
Let $U_1$, $U_2$ be $1$-dimensional homogeneous analytic quantum walks.
Suppose that there exists a non-zero uniform intertwiner from $U_1$ to $U_2$.
Then there exists a quantum walk $U$ such that
\begin{itemize}
\item
$U$ is a constant walk or an indecomposable quantum walk,
\item
there exists an analytic isometric intertwiner from $U$ to $U_1$,
\item
and there exists an analytic isometric intertwiner from $U$ to $U_2$.
\end{itemize}
\end{corollary}

Taking direct sum of $U$ in the corollary, we obtain ``the greatest common deviser'' of the two walks $U_1$ and $U_2$.

\section{Realizability and non-realizability by continuous-time quantum walks}
\label{section: realizability by CTQW}

\begin{theorem}\label{theorem: realizability by a CTQW}
Let $U = (\Hil, (U^t)_{t \in \Z}, E, \rho)$ be an arbitrary $1$-dimensional homogeneous analytic quantum walk.
The following conditions are equivalent:
\begin{enumerate}
\item
For every eigenvalue function, the winding number is zero.
\item
There exists an analytic homogeneous continuous-time quantum walk $(\Hil, (V^{(t)})_{t \in \R}, E, \rho)$ satisfying $V^{(1)} = U$.
\item
There exists a uniform continuous-time quantum walk 
$\left(\Hil, \left( V^{(t)} \right)_{t \in \R}, E \right)$ satisfying $V^{(1)} = U$.
\item
There exists a uniform continuous-time quantum walk $\left(\widetilde \Hil, \left(\widetilde V^{(t)}\right)_{t \in \R}, \widetilde E \right)$ and uniform unitary intertwiner $W$ from the discrete-time quantum walk $\left(\widetilde \Hil, \left(\widetilde V^{(t)}\right)_{t \in \Z}, \widetilde E \right)$ to $U$.
\end{enumerate}
\end{theorem}

If the condition $(3)$ holds, since the equation $V^{(t)} U = U V^{(t)}$ holds,
the family $(V^{(t)})_{t \in \R}$ gives a $1$-parameter group of uniform isomorphisms from the object $U$ to $U$.
Now we have already determined the set of uniform intertwiners from $U$ to $U$.
We are ready to prove the above theorem.

\begin{proof}
The implication $(3) \Rightarrow (4)$ is trivial.
Suppose $(4)$. Defining $V^{(t)}$ by $W \widetilde{V}^{(t)} W^{-1}$ we obtain $(3)$.
The conditions $(1)$, $(2)$, $(3)$ are equivalent by Theorem 4.2 in \cite{SakoIntertwiner}.

To get the idea of the proof for Theorem 4.2 in \cite{SakoIntertwiner},
let us consider the spectial case that $U$ is the indecomposable model quantum walk $\wU_{\lambda, p} = \left( L^2(\T_p), (M[\lambda]^t)_{t \in \Z}, \frac{1}{\ii}\frac{d}{dk} \right)$.
Recall that $p$ is the minimal period of $\lambda$ by Lemma \ref{lemma: decomposition of model QW}.
Suppose that the winding number of $\lambda$ is $0$.
Then there exists an analytic function $h \colon \R \to \R$ with period $p$ such that $\exp(\ii h) = \lambda$.
The continuous-time quantum walk $\left( L^2(\T_p), (M[\exp(\ii t h)])_{t \in \R}, \frac{1}{\ii}\frac{d}{dk} \right)$ realizes the discrete-time quantum walk $\wU_{\lambda, p}$.

Conversely, suppose that there exists a continuous-time uniform quantum walk
$( L^2(\T_p)$, $(V^{(t)})_{t \in \R}$, $\frac{1}{\ii}\frac{d}{dk} )$ 
satisfying that $V^{(1)} = M[\lambda]$.
Then for every real number $t$, $V^{(t)}$ is an intertwiner between 
$\wU_{\lambda, p}$ and $\wU_{\lambda, p}$.
By Theorem \ref{theorem: the set of intertwiners},
for every $t$ there exists a unique continuous function $f_t \colon \R \to \T_p$
with period $p$ satisfying that $V^{(t)} = M[f_t]$.
The family $\{f_t\}_{t \in \R}$ of continuous functions satisfies that
\[M[f_s f_t] = M[f_s] M[f_t] = V^{(s)} V^{(t)} = V^{(s+ t)} = M[f_{s + t}].\]
It follows that for every $s, t \in \R$, we have $f_s f_t = f_{s + t}$.
Let us denote by $w(f)$ the winding number of a continuous function $f \colon \T_p \to \R$.
We have $w(f_s) + w(f_t) = w(f_s f_t) = w(f_{s + t})$.
This means that the mapping $\R \ni t \mapsto w(f_t) \in \Z$ is a group homomorphism.
A group homomorphism from $\R$ to $\Z$ should be the $0$-map.
It follows that $w(\lambda) = w(f_1) = 0$.
\end{proof}

\begin{example}
\label{example : existence of CTQW}

Let $a, b$ be real numbers satisfying $0 < a < 0$, $a^2 + b^2 = 1$.
The walk
$\left(
\begin{array}{cc}
a S^{-1} & - b S^{-1}\\
b S  & - a S
\end{array}
\right)$
can be realized by a homogeneous analytic continuous-time quantum walk.
Indeed, the winding numbers of the eigenvalue functions
\[a \cos k \pm \ii \sqrt{1 - a^2 \cos^2 k} \]
are zero.

The walk $\left(
\begin{array}{cc}
a S & -b S\\
b & a
\end{array}
\right)$ can not be realized by a uniform continuous-time quantum walk.
Indeed, the winding number of the eigenvalue function
\[\lambda(k) 
= \exp (\ii k /2) \left( a \cos (k / 2)
- \ii \sqrt{1 - a^2 \cos^2 (k /2)} \right)\]
is $1$.
\end{example}

\begin{example}
The $1$-dimensional $3$-state Grover walk 
can be realized by a homogeneous analytic continuous-time quantum walk.
Indeed, the complete list of the eigenvalue functions is
\[\lambda_1(k) = -1, \quad \lambda_2 = \frac{2 + \cos k}{3} + \frac{1}{3} \ii \sin \frac{k}{2} \sqrt{10 + 2 \cos k}.\]
The winding numbers are $0$

The $1$-dimensional $4$-state Grover walk can not be realized by a uniform continuous-time quantum walk.
Indeed, the function
\[\lambda_+(k) = \dfrac{\cos k + \cos 3k}{2} + \ii \sin k \sqrt{1 + 4 \cos^4 k}\]
is one of the eigenvalue functions.
The winding number is $1$.
\end{example}

\section{concluding remark: how useful the category of quantum walks is.}

The category preserves a language.
By this language, we are able to say the following, for example.
\begin{enumerate}
\item
For a discrete-time quantum walk $U$,
its
realizability by a continuous-time quantum walk is equivalent to existence of $1$-parameter group of isomorphisms from $U$ to $U$.
\item
Given two walks $U_1$, $U_2$,
we can ask whether the dynamical system described by $U_1$ can be executed by $U_2$ or not.
This is equivalent to existence of an isometric morphism from the object $U_1$ to the object $U_2$,
\item
The inverse Fourier transform of a homogeneous walk is also an object.
Many researchers have regarded it as an intermediate product.
\end{enumerate}

\bibliographystyle{amsalpha}
\bibliography{category_of_QW.bib}

\end{document}